%% file: main.tex
\PassOptionsToPackage{nameinlink,capitalise}{cleveref}
\documentclass[a4paper, USenglish, cleveref, autoref, thm-restate, dvipsnames]{lipics-v2021}

\pdfoutput=1 
\hideLIPIcs  

\nolinenumbers

\title{A Regular and Complete Notion of Delay
for Streaming String Transducers} 

\author{Emmanuel Filiot}
{Universit\'e libre de Bruxelles, Belgium}
{efiliot@ulb.ac.be}
{https://orcid.org/0000-0002-2520-5630}
{Emmanuel Filiot is a senior research associate at F.R.S-FNRS.}

\author{Isma\"el Jecker}
{University of Warsaw, Poland}
{ismael.jecker@gmail.com}
{https://orcid.org/0000-0002-6527-4470}
{Supported by the ERC grant INFSYS, agreement no. 950398.}

\author{Christof Löding}
{RWTH Aachen University, Germany}
{loeding@cs.rwth-aachen.de}
{https://orcid.org/0000-0002-1529-2806}
{}

\author{Sarah Winter}
{Universit\'e libre de Bruxelles, Belgium}
{swinter@ulb.ac.be}
{https://orcid.org/0000-0002-6527-4470}
{Sarah Winter is a postdoctoral researcher at F.R.S.-FNRS.}

\authorrunning{E. Filiot, I. Jecker, C. L\"oding and S. Winter} 

\funding{This work was partially supported by the Fonds de la Recherche Scientifique – F.R.S.-FNRS under the MIS project F451019F.}

\Copyright{Emmanuel Filiot, Isma\"el Jecker, Christof L\"oding and Sarah Winter} 



\category{} 

\EventEditors{Petra Berenbrink, Mamadou Moustapha Kant\'{e}, Patricia Bouyer, and Anuj Dawar}
\EventNoEds{4}
\EventLongTitle{40th International Symposium on Theoretical Aspects of Computer Science (STACS 2023)}
\EventShortTitle{STACS 2023}
\EventAcronym{STACS}
\EventYear{2023}
\EventDate{March 7--9, 2023}
\EventLocation{Hamburg, Germany}
\EventLogo{}
\SeriesVolume{254}
\ArticleNo{17}

\usepackage{pict2e}
\usepackage{changepage}
\usepackage{xspace}
\usepackage{color}
\usepackage{mathtools}
\usepackage{tikz}
\usetikzlibrary{arrows,automata,trees,backgrounds,decorations.pathmorphing,decorations.pathreplacing,positioning,calc,shapes.geometric,fadings}
\tikzset{shorten >=1pt, >=stealth, auto, node distance=6em, initial text=}
\usepackage{xpatch}
\usepackage{apptools}
\usepackage{bbm}
\usepackage{stmaryrd}
\SetSymbolFont{stmry}{bold}{U}{stmry}{m}{n}
\usepackage{anyfontsize}

\input{macro}

\begin{document}

\maketitle

\begin{abstract}
The notion of delay between finite transducers
is a core element of numerous
fundamental results of transducer theory.
The goal of this work is to provide a similar notion
for more complex abstract machines:
we introduce a new notion of delay
tailored to measure the similarity between
streaming string transducers (SST).
We show that our notion is \emph{regular}: 
we design a finite automaton that can check whether
the delay between any two SSTs executions is smaller than some given bound.
As a consequence, our notion enjoys good decidability properties:
in particular,
while equivalence between non-deterministic SSTs
is undecidable, we show that 
equivalence \emph{up to fixed delay} is decidable.
Moreover, we show that
our notion has good \emph{completeness} properties:
we prove that two  SSTs are equivalent if and only if they are equivalent
up to some (computable) bounded delay. Together with the regularity of
our delay notion, it provides an alternative proof that SSTs equivalence is
decidable. Finally, the definition of our delay notion is
machine-independent, as it only depends on the \emph{origin semantics}
of SSTs. As a corollary, the completeness result also holds for
equivalent machine models such as deterministic two-way transducers,
or MSO transducers.
\end{abstract}

\newpage

\input{Sections/introduction}

\input{Sections/definitions}

\input{Sections/delay}

\section{Completeness and regularity}
\input{Sections/completness}

\input{Sections/regularity}

\input{Sections/applications}

\input{Sections/conclusion}

\bibliography{biblio}

\newpage

\appendix

\input{Appendix/discussion-delay-notion}

\input{Appendix/comparison}

\input{Appendix/completeness-details.tex}

\input{Appendix/regularity-details.tex}

\input{Appendix/inclusion}

\end{document}

%% file: macro.tex
\newcommand{\myparagraph}[1]{\vskip0.7\baselineskip\noindent\textcolor{black} {\sffamily\bfseries#1\enskip}} 

\global\long\def\st{\mid}

\renewcommand{\epsilon}{\varepsilon}

\newcommand{\alp}{\Sigma}
\newcommand{\alpo}{\Sigma}
\newcommand{\aut}{A}
\newcommand{\tra}{T}
\newcommand{\sst}{T}

\newcommand{\idsub}[1]{\textsf{Id}_{#1}}

\providecommand\lang{}
\renewcommand{\lang}[1]{L(#1)}
\newcommand{\rel}[1]{R(#1)}

\newcommand{\sub}[2]{{\mathcal S}_{#1,#2}}
\newcommand{\subs}{{\mathcal S}}
\newcommand{\var}{\mathcal{X}}
\newcommand{\varO}{O}
\newcommand{\delay}{\textsf{delay}}
\newcommand{\diff}{\textsf{diff}}

\newcommand{\cdelay}[1]{\mathbb{C}^{\dashv}_{k,\ell,#1}}

\newcommand{\outp}[1]{\textsf{out}(#1)}

\newcommand{\pump}[2]{\textsf{pump}_{#2}(#1)}

\newcommand{\neighS}{j_1}
\newcommand{\neighT}{j_2}

\newcommand{\origin}[1]{\textsf{ori}_{#1}}
\newcommand{\weight}{\textsf{weight}}
\newcommand{\maxdiff}{\textsf{max-diff}}
\newcommand{\oriout}[1]{\mathsf{\tilde{out}}(#1)}

\newcommand{\prefix}{{\alpha}}
\newcommand{\infix}{{\beta}}
\newcommand{\suffix}{{\gamma}}
\newcommand{\proot}{\textsf{root}}
\newcommand{\cut}{\textsf{cut}}
\newcommand{\nextcut}{\textsf{next-cut}}

\newcommand{\fact}{\textsf{factors}}

\newcommand{\toNFT}{\textsf{NFT}}

\newcommand{\N}{\mathbbm{N}}

\newcommand{\SST}{SST\xspace}
\newcommand{\SSTs}{SSTs\xspace}
\newcommand{\NFT}{NFT\xspace}

\newcommand{\PSPACE}{{\upshape\textsc{PSpace}}\xspace}

\newcommand{\mdiff}{\textnormal{\textsf{MD}}}
\newcommand{\mdiffe}{\textnormal{\textsf{eMD}}}
\newcommand{\mdiffd}{\textnormal{\textsf{dMD}}}

\makeatletter 
\xpatchcmd{\thmt@restatable}
{\csname #2\@xa\endcsname\ifx\@nx#1\@nx\else[{#1}]\fi}
{\IfAppendix{\csname #2\@xa\endcsname}{\csname #2\@xa\endcsname\ifx\@nx#1\@nx\else[{#1}]\fi}}
{}{} 
\makeatother

\usepackage{todonotes}

\newcommand{\manu}[2][]{}
\newcommand{\isma}[2][]{}
\newcommand{\christof}[2][]{}
\newcommand{\sarah}[2][]{}


%% file: Sections/introduction.tex
\section{Introduction}\label{sec:intro}

Transducers are
another fundamental extension of finite automata for computing \emph{functions}, and more generally \emph{relations},
between structures. In this paper, we consider string-to-string
transducers, which operate on (input) strings and produce (output) strings. 
The most basic, \emph{finite transducers}, are obtained by adding output words on the transitions of a
finite
automaton~\cite{DBLP:journals/iandc/Schutzenberger61a,DBLP:journals/ibmrd/ElgotM65}. 
At the heart of several important results in the theory of finite transducers is a 
notion, called \emph{delay}, allowing to finely compare, in a regular
way (with a finite automaton), transducer executions. The goal of this
paper is to provide such a notion for a strictly more powerful class
of transducers, streaming string
transducers~\cite{DBLP:conf/fsttcs/AlurC10}, which have received a lot of
attention in the recent years, due to their robustness and equivalence
with many other formalisms to define string-to-string
functions. In particular, this paper answers positively the following
high-level question: 
    \smallskip
\begin{adjustwidth}{3pt}{3pt}
\emph{Is it possible to  measure in a regular manner how differently executions of
  equivalent streaming string transducers produce their outputs?}
\end{adjustwidth}
\smallskip
\myparagraph{Finite transducers.} We first explain how the above question is answered for finite transducers.
Transitions of finite transducers over some alphabet
$\Sigma$ are tuples
$(p,\sigma,w,q)$ where $p,q$ are states, $\sigma\in\Sigma$ is a symbol read on
the input tape, and $w\in\Sigma^*$ is a word (possibly empty) written on the
output tape. A finite transducer execution, i.e., a sequence of
successive transitions $\rho =
(p_1,\sigma_1,w_1,p_2)\dots
(p_{n-1},\sigma_n,w_n,p_n)$, operates on the input word
$\textsf{in}(\rho) = \sigma_1\dots\sigma_n$ and produces the output
word $\textsf{out}(\rho) = w_1\dots w_n$. The semantics of a
finite transducer is the set of pairs
$(\textsf{in}(\rho),\textsf{out}(\rho))$ for all accepting executions
$\rho$. Two different executions $\rho_1$ and $\rho_2$ might define
the same input/output pair, but can produce the output in a different way because the output words
$w_i$ occurring on the transitions of the two runs may differ, although
their whole concatenation is the same.

\begin{figure}[t]
\begin{tikzpicture}[baseline,thick,scale=1,
  every path/.style={shorten <=0cm,shorten >=0cm}]
\pgfdeclarelayer{bg}    
\pgfsetlayers{bg,main}
\def\H{0.3}
\def\W{0.25}
\def\Hx{2.7}
\def\Hy{0.5}

\node[anchor=west] at (-0.85+\Hx,0.4*\H){\small $\rho_5$ :\strut};
\node[anchor=west] at (-0.85+\Hx,\Hy + 0.4*\H){\small $\rho_4$ :\strut};
\node[anchor=west] at (-0.85+\Hx,2*\Hy + 0.4*\H){\small $\rho_3$ :\strut};
\node[anchor=west] at (-0.85+\Hx,3*\Hy + 0.4*\H){\small $\rho_2$ :\strut};
\node[anchor=west] at (-0.85+\Hx,4*\Hy + 0.4*\H){\small $\rho_1$ :\strut};

\foreach \y/\i/\c in {  0/0/4,0/1/3,0/2/2,0/3/1,0/4/1,0/5/2,0/6/3,0/7/4,
                        1/0/1,1/1/2,1/2/3,1/3/4,1/4/4,1/5/3,1/6/2,1/7/1,
                        2/0/4,2/1/4,2/2/3,2/3/3,2/4/2,2/5/2,2/6/1,2/7/1,
                        3/0/1,3/1/1,3/2/2,3/3/2,3/4/3,3/5/3,3/6/4,3/7/4,
                        4/0/2,4/1/2,4/2/2,4/3/2,4/4/4,4/5/4,4/6/4,4/7/4}{

\draw[black!20!white] (\i*\W+5*\Hx,\y*\Hy) -- (\i*\W+5*\Hx,\H+\y*\Hy);

\ifthenelse{\i=0}
{\draw[thick] (5*\Hx,\y*\Hy) -- (8*\W+5*\Hx,\y*\Hy) -- (8*\W+5*\Hx,\H+\y*\Hy) -- (5*\Hx,\H+\y*\Hy) -- (5*\Hx,\y*\Hy);}
{}
\node at (\i*\W+5*\Hx+0.5*\W,\y*\Hy + 0.4*\H){\scriptsize \c \strut};

\foreach \x in {1,2,3,4}{
\ifthenelse{\y=0 \AND \i=0}
    {\node at (4*\W+\x*\Hx,-0.7*\Hy){\small $t=\x$};}
    {}
    
\ifthenelse{\x>1 \AND \i=0}{
    \draw[-stealth] (\x*\Hx-0.5,\y*\Hy+0.5*\H) -- (\x*\Hx-0.2,\y*\Hy+0.5*\H);}
    {}
    
\ifthenelse{\i=0}{
  \draw[thick] (\x*\Hx,\y*\Hy) -- (8*\W+\x*\Hx,\y*\Hy) -- (8*\W+\x*\Hx,\H+\y*\Hy) -- (\x*\Hx,\H+\y*\Hy) -- (\x*\Hx,\y*\Hy);}
  {}

\ifthenelse{\i<4 \AND \(\x>\c \OR \x=\c\)}
  {\draw[fill=white!75!ForestGreen]
  (\i*\W+\W+\x*\Hx,0+\y*\Hy) rectangle (\i*\W+\x*\Hx,\H+\y*\Hy);
  \draw (\i*\W+\x*\Hx,0+\y*\Hy) -- (\i*\W+\x*\Hx,\H+\y*\Hy);
  \draw (\i*\W+\W+\x*\Hx,0+\y*\Hy) -- (\i*\W+\W+\x*\Hx,\H+\y*\Hy);
  \node at (\i*\W+\x*\Hx+0.5*\W,\y*\Hy + 0.4*\H){\scriptsize $a$\strut};}
  {
  \ifthenelse{\i>3 \AND \(\x>\c \OR \x=\c\)}
  {\draw[fill=white!75!Blue]
  (\i*\W+\W+\x*\Hx,\y*\Hy) rectangle (\i*\W+\x*\Hx,\H+\y*\Hy);
  \draw (\i*\W+\x*\Hx,\y*\Hy) -- (\i*\W+\x*\Hx,\H+\y*\Hy);
  \draw (\i*\W+\W+\x*\Hx,\y*\Hy) -- (\i*\W+\W+\x*\Hx,\H+\y*\Hy);
  \node at (\i*\W+\x*\Hx+0.5*\W,\y*\Hy + 0.4*\H){\scriptsize $b$\strut};}
  {\begin{pgfonlayer}{bg}
  \draw[black!20!white] (\i*\W+\x*\Hx,\y*\Hy) -- (\i*\W+\x*\Hx,\H+\y*\Hy);
  \draw[black!20!white] (\i*\W+\W+\x*\Hx,\y*\Hy) -- (\i*\W+\W+\x*\Hx,\H+\y*\Hy);
  \end{pgfonlayer}}
  }
}
}

\end{tikzpicture}
  \caption{Five ways of producing the output
  sequence $a^4b^4$, and the corresponding origin functions.}\label{fig:delay}
\end{figure}
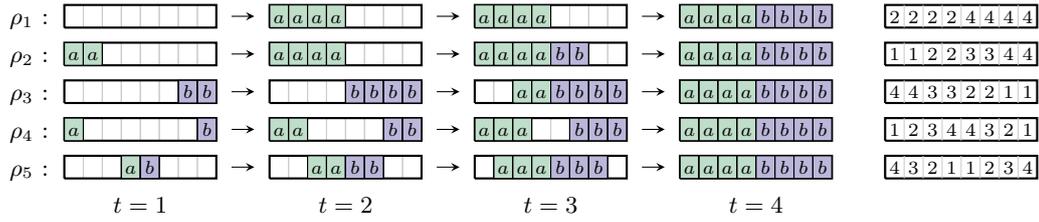

\myparagraph{Origin information.} Differences in the way
transducers produce their output are best captured by the notion of
\emph{origin information}, initially proposed for streaming
string transducers~\cite{DBLP:conf/icalp/Bojanczyk14}. An
origin function maps any position of the output word to the input
position where it was produced. In an execution such as $\rho$ above, any
position of $w_i$ has origin $i$. \cref{fig:delay}
illustrates five different ways of producing the same output word
$a^4b^4$ (the input word of length 4 is irrelevant here and not depicted).
Consider for example the execution $\rho_1$. When reading the first
input symbol (timestep $t=1$), nothing is produced. At timestep $t=2$,
$a^4$ is produced, and its positions have origin $2$, as depicted to
the right of the figure. The sequence $b^4$ is produced at timestep
$t=4$ so its corresponding positions have origin $4$. Note that
$\rho_3,\rho_4,\rho_5$ do not correspond to finite
transducer executions, as $a^4b^4$ is not produced from
left to right.

\myparagraph{Delay for finite transducers.} A natural way of comparing two finite transducer executions
$\rho_1,\rho_2$ on the same input and producing the same output, is to
compare the origin functions $o_1,o_2$ they induce, and in particular 
how much one is ahead of the other. In \cref{fig:delay},
both $\rho_1$ and $\rho_2$ produce the same output
from left to right, yet at different speed:
$\rho_1$ produces bigger chunks of output at lower frequency.
The \emph{delay} between both executions is $2$ for $t \in \{1,3\}$
and $0$ for $t \in \{2,4\}$.
The global delay $\delay(\rho_1,\rho_2)$
is defined as the maximal delay along the executions,
in this case $2$. Two transducers $T_1$ and $T_2$ are said to be
\emph{$k$-delay equivalent} if for each run $\rho_1$ of $T_1$ or $T_2$
the other transducer has a run $\rho_2$
with the same input and output satisfying~$\delay(\rho_1,\rho_2) \leq k$.
This delay notion enjoys two important properties:    
\begin{description}
    \item[Regularity:] While the set $E = \{ (\rho_1,\rho_2)\mid
      \textsf{in}(\rho_1)=\textsf{in}(\rho_2)\wedge
      \textsf{out}(\rho_1)=\textsf{out}(\rho_2)\}$ is not
      regular\footnote{Regularity here is defined as classical
      regularity of word languages, where a pair $(\rho_1,\rho_2)$,
      are encoded as a single word over a product alphabet of
      transitions.} in
      general for finite transducer executions, its restrictions to
      pairs of executions with bounded delay is regular: For every $k \in \mathbbm{N}$, the set $
      E_k = \{ (\rho_1,\rho_2)\in E\mid\delay(\rho_1,\rho_2)\leq k\}
      $
      is regular~\cite{DBLP:conf/icalp/FiliotJLW16}. As finite automata compose well with other
    abstract machines, this allows to use the delay
    in all kinds of constructions
    while preserving 
    good decidability properties. For instance $k$-delay equivalence
    is decidable~\cite{DBLP:conf/icalp/FiliotJLW16}.
  \item[Completeness:] For any two finite transducers $T_1,T_2$
    defining string-to-string \emph{functions}, there exists a
    computable bound $k$ such that $T_1$ and $T_2$ are equivalent
    iff they are $k$-delay equivalent.
    The completeness even holds in broader classes~\cite{DBLP:conf/icalp/FiliotJLW16},
    such as \emph{finite-valued} transducers,
    for which there is a global bound on the number of outputs mapped
    to a single input.
\end{description}


\myparagraph{Applications.}
Due to its regularity and completeness,
the notion of delay, while basic,
proves to be very powerful,
and is omnipresent in the study of finite transducers. It provides decidable approximations
    of various undecidable decision
    problems, which reveals to be exact
for broad classes of transducers,
    such as finite-valued
    transducers~\cite{DBLP:conf/icalp/FiliotJLW16}. Various patterns characterizing
    important subclasses of transducers are based on the delay notion:
    For instance, transducers which are \emph{sequential}~\cite{DBLP:journals/tcs/Choffrut77,DBLP:journals/tcs/BealCPS03} (i.e., definable by an input-deterministic transducer);
    equivalent to finite union of sequential functions~\cite{DBLP:conf/dlt/JeckerF15,DBLP:conf/fossacs/DaviaudJRV17};
    finite-valued~\cite{DBLP:conf/mfcs/SakarovitchS08,DBLP:journals/mst/SakarovitchS10}. Moreover,
    it has been used to show that any finite-valued transducer can be
    decomposed as a union of 1-valued
    transducers~\cite{DBLP:journals/ita/Weber96,DBLP:conf/stacs/SakarovitchS08}. Canonical notions
    for finite transducers are also based on delay: for
    input-deterministic
    transducers~\cite{DBLP:journals/tcs/Choffrut03} and $1$-valued
    transducers~\cite{DBLP:journals/siamcomp/ReutenauerS91}.

\myparagraph{Streaming string transducers and regular functions.} 
\emph{Streaming string transducers} (SST)
are obtained by equipping deterministic finite automata with a finite
set of registers that store output~\cite{DBLP:conf/fsttcs/AlurC10}.
Those registers cannot be tested, but are updated by
concatenating them, or prepending/appending some symbols to them. E.g., consider a single-state SST with a loop which, whatever it
reads as input, updates its single register $O$ by the operation $O
:= aOb$, and finally outputs the whole content of $O$. The execution
$\rho_5$ in \cref{fig:delay} represents an execution of this
SST. Consider an equivalent single-state SST with two registers $X,Y$ which, whatever the input, executes  $X := Xa$ and $Y:=bY$, and
finally outputs the concatenation $XY$ ($\rho_4$ in
\cref{fig:delay} produces the output in this way). While
equivalent, those two transducers are not equivalent if the origin
information is included in the semantics.

SSTs define a robust set of string-to-string functions,
called \emph{regular functions}, which can be equivalently defined by
deterministic two-way transducers~\cite{DBLP:conf/fsttcs/AlurC10},  
monadic second-order
transductions~\cite{DBLP:books/daglib/0030804,DBLP:journals/tocl/EngelfrietH01,DBLP:conf/icalp/AlurD11},
regular transducer
expressions~\cite{DBLP:conf/csl/AlurFR14,DBLP:journals/ijfcs/BaudruR20,DBLP:conf/lics/DartoisGGK22,DBLP:conf/lics/DartoisGK21},
and a logic with origins~\cite{DBLP:conf/lics/DartoisFL18}. Regular
functions also enjoy a Krohn-Rhodes decomposition
theorem~\cite{DBLP:journals/corr/abs-1810-08760}. SSTs have a decidable
equivalence problem~\cite{DBLP:journals/siamcomp/Gurari82}, have been applied
to the verification of \emph{list-processing programs}~\cite{DBLP:conf/popl/AlurC11}, and have been implemented for evaluating string transformations~\cite{DBLP:conf/popl/AlurDR15}.

There are two natural ways of extending the delay defined for
finite transducers to compare
executions of SSTs with the same input and output:
we can either simply compare the number of produced output letters
without caring about the positions where it is produced,
or we can count the number of positions whose
output has already been produced in one run
but not in the other.
Both methods match the previously defined delay
if the output is produced from left to right,
but unfortunately neither preserves
conjointly the regularity and the completeness
once other output production mechanisms, such as in SSTs, are considered\footnote{We refer the
  interested reader to the appendix for more details about those
  natural extensions and an explanation of why they are not
  satisfactory.}. As the basic notions of delay for SSTs
fail to satisfy good properties,
better ways of comparing SSTs were introduced,
yet none that conjointly has good decidability
\emph{and} completeness properties.
For instance, \emph{bounded regular resynchronizers} can
alter origin functions of SSTs in a controlled manner
while preserving good decidability properties,
but they are are not complete~\cite{DBLP:journals/corr/abs-1807-08053,DBLP:phd/hal/Bose21}.
In the restricted setting of \emph{single-register} SSTs,
a regular and complete notion of delay is defined based on word equations~\cite{DBLP:conf/stacs/GallotMPS17}.
This approach was applied to prove a decomposition theorem
for single-register SSTs, but unfortunately fails  when more registers
are allowed.

\myparagraph{Contributions.} We define a delay notion for comparing SST
executions, based
on the origin functions they induce,  that is both regular (\cref{thm:trueDelayResync}) and
complete for fundamental SST decision problems including equivalence
(\cref{thm:completeFunction,cor:completeMinim,cor:completeRat}). This
delay notion is described in the next subsection. We first present
our results. Since
our notion is based on origins, it more generally applies to regular
functions with origin information. SSTs are known to be
equivalent to deterministic two-way and MSO transducers,
via origin-preserving encodings~\cite{DBLP:conf/icalp/Bojanczyk14}, so
we obtain as a corollary
that our delay notion is also complete for equivalence of
deterministic two-way transducers and MSO transducers (\cref{cor:general}).

The origin semantics allows us to recover decidability of several decision problems:
While non-deterministic SST equivalence is undecidable,
we can decide whether two non-deterministic SSTs define the same relations \emph{with
  same origins}. However, asking for identical origins is very
restrictive: it fails to detect equivalence if the origins are
perturbed ever so slightly. This observation was already made
in~\cite{DBLP:conf/icalp/FiliotJLW16} in the context of finite transducers, where it is proposed to relax several decision problems,
such as equivalence with same origins, to decision problems with
\emph{close} origins, such as equivalence up to a given delay bound
$k$. We prove that the inclusion, equivalence and variable
minimization problems up to a given delay bound are decidable for
(non-deterministic) SSTs
(\cref{thm:complexityInclusion,thm:delayVarMin}), while in general
those problems are undecidable. Since our delay notion is complete for
equivalence of (deterministic) SSTs, the latter result provides an
alternative proof that such SSTs have decidable equivalence problem,
and sheds light on the intrinsic reasons why SSTs executions are
equivalent.

\myparagraph{Delay for SSTs.} Our notion of delay is based on the following observation:
The order of production of each output segment
that is a power of a small word
should have no impact on the delay.  For every $\ell \in \mathbbm{N}$,
we introduce the measure $\delay_\ell$,
which first decomposes the output
into blocks that are powers of words smaller than or equal to $\ell$.
Then, at any position $j$ ending a block (instead of \emph{any} position), we
measure the maximal difference at any timestep $t$, between the number
of output positions at the left of $j$ produced before timestep $t$ in
the first and second executions. We then take the maximal such value
for all~$j$. Consider $\rho_4$ and $\rho_5$ for example. $\delay_1$ splits the output $a^4b^4$ of \cref{fig:delay}
into two blocks, $a^4$ and $b^4$. Consider ending block position
$j=4$. At any timestep, the maximal difference of quantity of outputs produced at
the left of $j$ is always $0$, because $\rho_4,\rho_5$ produce
exactly one symbol per timestep before position $j=4$. The same
reasoning applies for ending block position $j=8$, where here instead,
exactly two symbols are produced by $\rho_4,\rho_5$ at any
timestep, at the left of position $j=8$. Hence $\delay_1(\rho_4,\rho_5)=0$.  We refer to \cref{sec:resync} for a
smooth introduction to our delay notion and the main results and to
\cref{subsec:trueDelayResyncComplete} for a proof of its
completeness, and to \cref{subsec:trueDelayResync} for a proof of its regularity. The completeness proof is perhaps the most involved. It
is based on a key pumping lemma (\cref{lem:technical}), which
intuitively states that for any SST, there exist computable
bounds $k$ and $\ell$ such that, if the delay $\delay_\ell$ between
two executions on the same input/output is greater than $k$, then those two executions can be
pumped to construct two executions over the same input but with
\emph{different} outputs.

%% file: Sections/definitions.tex
\section{Preliminaries}\label{sec:prel}

We fix our notation.
Let $\N$ denote the set of non-negative integers.
The interval between integers $a$ and $b$ including resp.\ excluding $a$ and $b$ is denoted $[a,b]$ resp.\ $(a,b)$.

\myparagraph{Free monoid.}
Given a finite alphabet $\alp$,
the \emph{free monoid} over $\alp$ is
the monoid $(\alp^*,\cdot, \epsilon)$ defined as follows.
The set $\alp^*$ is composed of finite sequences of elements of $\alp$, called \emph{words}.
The operation $\cdot$ is the usual word concatenation. The neutral element $\epsilon$ is the empty word.

A subset $L \subseteq \alp^*$ is called a \emph{language}. 
Given a word $u = a_1\cdots a_n\in \alp^*$, $|u|$ denotes its length,
$u[i]$ denotes its $i$th letter $a_i$, $u[i,j]$ denotes its infix
$a_i\cdots a_j$.
Given an additional word $v = b_1\cdots b_n \in \alp^*$, $u \otimes v$ denote the \emph{convolution} $\binom{a_1 \cdots a_n}{b_1 \cdots b_n} \in (\alp \times \alp)^*$.

\myparagraph{Automaton.}
A \emph{(finite state) automaton} is a tuple $\aut = (\alp,Q,I,\Delta,F)$
composed of a finite alphabet $\alp$,
a finite set of states $Q$,
a set of initial states $I \subseteq Q$,
a set of final states $F \subseteq Q$,
and a set of transitions
$\Delta \subseteq Q \times \alp \times Q$.
A \emph{run} of $\aut$
is a sequence $\rho = q_0 a_1 q_1 a_2 q_2 \cdots a_n q_n \in Q (\alp Q)^*$
such that $(q_{i-1}, a_i, q_i) \in \Delta$ for every $1 \leq i \leq n$.
The \emph{input} of $\rho$ is the word $a_1a_2 \cdots a_n \in \alp^*$.
The run $\rho$ is called \emph{accepting} if its first state is initial ($q_0 \in I$)
and its last state is final ($q_n \in F$).
The automaton $\aut$ is \emph{deterministic} if $I$ is a singleton and
$\Delta$ is expressed as a function $\delta: Q \times \alp \rightarrow Q$.
The \emph{language recognized} by $\aut$ is the set $\lang{\aut} \subseteq \alp^*$ composed of the inputs of all its accepting runs.

\myparagraph{Substitutions monoid.}
Given a finite alphabet $\alp$, a finite set $\var =
\{X_1,X_2,\ldots,X_n\}$ of \emph{variables}, and a designated
\emph{output variable} $\varO \in \var$, 
the \emph{(copyless) substitutions monoid} $(\sub{\var}{\alp},\circ,\idsub{\var})$
is defined as follows.
The set $\sub{\var}{\alp}$
contains all functions
$
\sigma: \var \rightarrow (\var \cup \alp)^*
$
such that no variable appears twice in the image of a variable
and no variable appears in the images of two distinct variables, i.e.,
the word $\sigma(X_1)\dots \sigma(X_n)$ does not contain twice the
same variable.
The composition $\sigma_1 \circ \sigma_2$
of two substitutions $\sigma_1,\sigma_2 \in \sub{\var}{\alp}$ is obtained by first applying $\sigma_2$, and then $\sigma_1$. Formally, it is the function $\hat{\sigma}_1(\hat\sigma_2(\cdot))$,
where each substitution $\sigma$ is morphically extended to $\hat \sigma$ over $(\var\cup \alp)^*$ by letting $\hat{\sigma}(X) = \sigma(X)$ and $\hat{\sigma}(\alpha) = \alpha$ for $\alpha\in \alp$.
We write $\sigma$ in place of $\hat\sigma$.
For instance, the substitution composition $(\sigma_1\colon X \mapsto \varepsilon) \circ (\sigma_2\colon X \mapsto aX) \circ (\sigma_3\colon X \mapsto bXc)$
maps $X$ to $bac$, since $\sigma_3(X) = bXc$, $\sigma_2(bXc) = baXc$ and $\sigma_1(baXc) = bac$.
%
The neutral element $\idsub{\var}$ is the identity function mapping
each variable $X \in \var$ to itself.
Finally, we denote by $\sigma_{\varepsilon}$ the substitution mapping
each variable $X \in \var$ to $\epsilon$.
We recall that $\varO$ is the designated output variable; the \emph{output} of a substitution $\sigma$ is the word $ \outp{\sigma} = (\sigma_{\varepsilon} \circ \sigma)(\varO)$.

\myparagraph{Streaming string transducer.}
A \emph{streaming string transducer} (\SST{}) is a tuple $\sst = (\alp, Q, I, \delta, F, \var, \varO, \kappa, \kappa_F)$
where $\aut_{\sst} = (\alp, Q, I, \delta, F)$ is a deterministic
automaton, called \emph{underlying automaton} of $\sst$,
$\var$ is a finite set of variables (also called registers),
$\varO \in \var$ is a final variable,
$\kappa\colon \delta \rightarrow \sub{\var}{\alpo}$
is an output function, and $\kappa_F\colon Q \rightarrow \sub{\var}{\alpo}$ is a final output function.
A \emph{run} of $\sst$ is a run $\rho = q_0 a_1 q_1 a_2 q_2 \cdots a_n q_n \in Q (\alp Q)^*$
of $\aut_{\sst}$,
we define $\kappa(\rho)$ as the substitution obtained by composing
sequentially all
substitutions occurring on the transitions and the final substitutions,
i.e. $\kappa(\rho) = \kappa((q_0,a_1,q_1)) \circ \kappa((q_1,a_2,q_2)) \circ \cdots \circ \kappa((q_{n-1},a_n,q_n)) \circ \kappa_F(q_n) \in \sub{\var}{\alpo}$.
The \emph{output} of the run $\rho$ is the word $(\sigma_{\varepsilon} \circ \kappa(\rho))(\varO)$.
The \emph{transduction} $\rel{\sst}$ \emph{recognized} (also called \emph{realized}) by $\sst$ is
the set of pairs $(u,v) \in \alp^* \times \alpo^*$
such that $\tra$ has an accepting run with input~$u$ and output~$v$.
Non-deterministic \SST{} are defined the same way except that $\aut_\sst$ is non-deterministic.
We emphasize that the transduction realized by a (deterministic) \SST{} is a partial function.

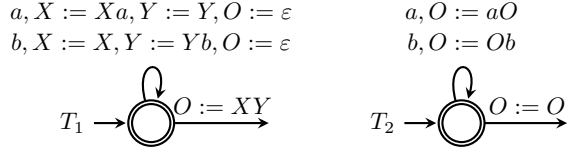
\begin{figure}[t]
\centering
\begin{tikzpicture}[baseline,thick,scale=1,every node/.style={scale=0.9},every loop/.style={looseness=10}]
  \tikzstyle{every state}+=[inner sep=3pt, minimum size=1.75em];
  \node[state, initial, initial text = $T_1$, accepting] (0) {};
  \coordinate[right of =0, xshift = -1em] (e);
  \draw[->] (0) edge[] node {$O := XY$} (e);
  \draw[->] (0) edge[loop above] node {$\begin{matrix} a, X := Xa, Y := Y, O := \varepsilon \\ b, X := X, Y := Yb, O := \varepsilon\end{matrix}$} ();
\end{tikzpicture}
    \qquad
    \begin{tikzpicture}[baseline,thick,scale=1,every node/.style={scale=0.9},every loop/.style={looseness=10}]
      \tikzstyle{every state}+=[inner sep=3pt, minimum size=1.75em];
      \node[state, initial, initial text = $T_2$, accepting] (0) {};

      \coordinate[right of =0, xshift = -1.5em] (e);
  
      \draw[->] (0) edge[] node {$O := O$} (e);
      \draw[->] (0) edge[loop above] node {$\begin{matrix} a, O := aO \\ b, O := Ob \end{matrix}$} ();
    \end{tikzpicture}
  \caption{Two deterministic \SST{}s that realize the same sorting function, cf.\ \cref{ex:intro}.}\label{fig:intro}
\end{figure}

\begin{example}\label{ex:intro}
    Depicted in \cref{fig:intro} are deterministic \SST{}s $T_1$ and
    $T_2$ that realize the same sorting function $f\colon \{a,b\}^*
    \to \{a,b\}^*$ that maps $u \in \{a,b\}^*$ to $a^mb^n$, where $m$ (resp.\ $n$) is the number of occurrences of $a$ (resp.\ $b$) in $u$.
\end{example}

\myparagraph{Origin information.}
A \emph{word with origins} is a word $\tilde u := u \otimes o \in (\Sigma \times \mathbbm{N})^*$ where each letter of $u$ is annotated with a natural number signifying its origin in time.
A \emph{transduction with origins} is a relation $R \subseteq \Sigma^* \times (\Sigma \times \mathbbm{N})^*$.

Naturally, we associate with a sequence $\lambda = \sigma_1 \sigma_2 \cdots \sigma_n \in \sub{\var}{\alp}^*$ of substitutions the output word with origins $\oriout{\lambda} = \outp{\lambda} \otimes i_1\cdots i_{|\outp{\lambda}|}$ with $i_j = \outp{\sigma_1' \sigma_2' \cdots \sigma_n'}[j]$ for $1 \leq j \leq |\outp{\lambda}|$ where for all $1 \leq t \leq n$, the substitution $\sigma_t' \in \sub{\var}{\N}$ is obtained by replacing each output letter in $\sigma_t$ with the number~$t$.
Furthermore, we associate with an \SST{} $T$ the \emph{transduction with origins} $R_{\origin{}}(T)$ that is the set of pairs $(u,\oriout{\lambda}) \in \Sigma^* \times (\Sigma \times \mathbbm{N})^*$ such that $T$ has an accepting run $\rho$ with input $u$ and sequence of substitutions $\kappa(\rho) = \lambda$.
For example, regarding the \SST{s} $T_1$ and $T_2$ from \cref{fig:intro}, we obtain that $(abaa,\binom{aaab}{1342}) \in R_{\origin{}}(T_1)$ and $(abaa,\binom{aaab}{4312}) \in R_{\origin{}}(T_2)$.


%% file: Sections/delay.tex
\section{Delay measure}\label{sec:resync}

As mentioned in the introduction, the concept of delay has proven to be a useful tool in the understanding of finite transducers.
Our goal is to introduce a robust delay measure suitable to gain a better understanding of streaming string transducers.
Towards that, we informally recall the notion of delay for finite transducers:
The delay between two finite transducer computations that produce the same output in the end, is a measure for how much ahead one output is compared to the other during the computation.
In other words, the difference between the length of the so-far produced output is measured.

A key difference between finite transducers and \SST{s} is that the
former build their outputs from left to right
while \SST{s} do not have this restriction.
To design a notion of delay taking this into account,
we use origin information to define a notion of \emph{weight difference}.
These notions are illustrated in \cref{ex:notions}.

\myparagraph{Weight difference.}
Given a word with origins $\tilde u := u \otimes o \in (\Sigma \times \mathbbm{N})^*$, a time $t \in \mathbbm{N}$, and $j \in \mathbbm{N}$,
we define the \emph{positional weight} $\weight_{j,t}(\tilde u) \in \N$ as the number of positions of $u$ up to~$j$ whose origin is no later than $t$, i.e.,
\[
\weight_{j,t}(\tilde u) = \big |\{ i \in \{1,2,\ldots, \text{min}(j,|u|)\} \st o[i] \leq t \}\big |.
\]
Note that $\weight_{j,t}(\tilde u) =
\weight_{|u|,t}(\tilde u)$ for all $j\geq
|u|$.
Given a second word with origins $\tilde v \in (\Sigma \times \mathbbm{N})^*$, and $j_1,j_2\in \mathbbm{N}$, we define the \emph{weight difference} as
\[
     \diff_{j_1,j_2,t}(\tilde u,\tilde v) =
     \big| \weight_{j_1,t}(\tilde u) - \weight_{j_2,t}(\tilde v) \big| \text{ and } \diff_{j,t}(\tilde u,\tilde v) = \diff_{j,j,t}(\tilde u, \tilde v).
\]
Furthermore, we define the \emph{maximal weight difference} as
\[
    \maxdiff_{j_1,j_2}(\tilde u,\tilde v) = \max_{t \in \N}
    \big(\diff_{j_1,j_2,t}(\tilde u,\tilde v)\big)  \text{ and } \maxdiff_j(\tilde u,\tilde v) = \max_{t \in \N}
\big(\diff_{j,t}(\tilde u,\tilde v)\big).
\]
We remark that the value of \maxdiff{} is bounded even though $t$ takes infinitely many values.
Also, we note that $\maxdiff_0(\tilde u,\tilde v)=0$.
We illustrate these notions.

\begin{example}\label{ex:notions}
We base our example on the \SSTs $T_1,T_2$ from \cref{fig:intro}.
On input $abaaa$ both \SSTs produce output $aaaab$.
The associated origins differ, we have $\tilde u = \binom{aaaab}{13452}$ and $\tilde v = \binom{aaaab}{54312}$.
We obtain $\diff_{2,0}(\tilde u,\tilde v) = 0$, $\diff_{2,1}(\tilde u,\tilde v) = 1$, $\diff_{2,2}(\tilde u,\tilde v) = 1$, $\diff_{2,3}(\tilde u,\tilde v) = 2$, $\diff_{2,4}(\tilde u,\tilde v) = 1$, $\diff_{2,5}(\tilde u,\tilde v) = 0$, and $\maxdiff_{2}(\tilde u,\tilde v) = 2$.
More generally, on input $aba^i$ the output is $a^{i+1}b$ and with origins we have 
$\tilde u_i = \left(\begin{smallmatrix} a & a & a & \cdots & a & a & b\\ 1 & 3 & 4 & \cdots & i+1 & i+2 & 2\end{smallmatrix}\right)$ and $\tilde v_i = \left(\begin{smallmatrix} a & a & \cdots & a & a & a & b\\ i+2 & i+1 & \cdots & 4 & 3 & 1 & 2\end{smallmatrix}\right)$ for all $i > 0$.
Moreover, $\maxdiff_{(i+1)/2} (\tilde u_i,\tilde v_i) = (i+1)/2$ for all odd $i$.
\end{example}

The first idea of a delay notion for \SST{s} similar to finite
transducers is to consider the maximal weight difference that can occur.
In~\cref{ex:notions}, since $T_1$ builds the $a$-output block from left to right and $T_2$ from right to left, their maximal weight difference is unbounded even though $T_1 \equiv T_2$.
Hence, this first idea of a delay notion violates the completeness requirement.
To avoid this problem,
we would like our notion of delay to reflect that,
for a periodic block, it is not important if a repetition of the period is appended or prepended:
the result is the same.
Therefore, we only measure the difference at the end of periodic blocks and not inside of them.
To this end, we introduce a new notion.

\myparagraph{Factors, cuts.}
The \emph{primitive root} of a word $u\in \alp^*$, denoted $\proot(u)$, is the shortest word $w$ such that $u = w^k$ for some positive integer $k$.
We call a word $u\in\alp^*$ \emph{primitive} if $\proot(u)=u$.

Let $u \in \alpo^*$ and $\ell > 0$.
We cut $u$ into factors such that each factor has a primitive root of length at most $\ell$.
The factors are chosen inductively from left to right in a way to maximize the size of each factor as follows.
The first factor $u_1$ of $u$ is the longest prefix of $u$ such that $|\proot(u_1)|\leq \ell$, the $i$th factor $u_i$ of $u$ is the longest prefix of $u'$ such that $|\proot(u_i)|\leq \ell$ where $u = u_1\cdots u_{i-1}u'$.
We refer to this unique \emph{$\ell$-factorization} as $\fact_\ell(u)$.
Moreover, we denote by $\cut_\ell(u)$ the set that contains the end positions of the factors referred to as \emph{$\ell$-cuts}.
For example, consider $u = aaababcbabaaaaa$ and $\ell = 2$, then the unique $\ell$-factorization is $aaa|ba|bc|baba|aaaa|$ and its $\ell$-cuts are $\{3,5,7,11,15\}$.
%



\myparagraph{Delay.} 
We define delay for a word and two origin annotations (of said word) by considering the maximal weight difference \emph{at the cut positions} which are obtained via the factorization into periodic words as introduced above.
Formally, given a word $w \in \Sigma^*$ and two annotated versions $\tilde u,\tilde v$ with origin, i.e., $\tilde u := w \otimes o_1,\tilde v := w \otimes o_2 \in (\Sigma \times \mathbbm{N})^*$, we define the \emph{$\ell$-delay} $\delay_{\ell}(\tilde u,\tilde v)$ (delay for short) as
\[
    \delay_{\ell}(\tilde u,\tilde v) =
    \max_{j \in \cut_{\ell}(w)} \maxdiff_j(\tilde u,\tilde v).
\]

Going back to \cref{ex:notions}, we have $\delay_{1}(\tilde u_i,\tilde v_i) = 0$, because the $1$-factorization of $a^{i+1}b$ is $a^{i+1}|b|$ and we have $\maxdiff_{i+1}(\tilde u_i,\tilde v_i) = \maxdiff_{i+2}(\tilde u_i,\tilde v_i) = 0$.




We apply the delay notion to transductions with origin.
Intuitively, two such transductions are ``close'' if for every pair $(u, w \otimes o_1)$ (from one transduction) there is some pair $(u, w \otimes o_2)$ (from the other transduction) such that the delay between these outputs with origins is small.
Let $R_1, R_2$ denote transductions with origin.
We say that $R_1$ is \emph{$(k,\ell)$-included} in $R_2$ (written $R_1 \subseteq_{k,\ell} R_2$) if the following holds: for all $(u,w \otimes o_1)\in R_1$,
there exists $(u,w \otimes o_2)\in R_2$
such that $\delay_{\ell}(w \otimes o_1,w \otimes o_2) \leq k$.
We say that $R_1$ is \emph{$(k,\ell)$-equivalent} to $R_2$ (written $R_1 \equiv_{k,\ell} R_2$) if $R_1 \subseteq_{k,\ell} R_2$ and conversely $R_2 \subseteq_{k,\ell} R_1$.
We are ready to state our first main result which illustrates the generality of our delay notion.

\begin{theorem}[Completeness]\label{thm:completeFunction}
Given two \SST{s} $T_1$ and $T_2$, there exist computable integers $k,\ell$ such that 
$T_1 \equiv T_2$ iff $R_{\origin{}}(T_1) \equiv_{k,\ell} R_{\origin{}}(T_2)$.
\end{theorem}

\cref{subsec:trueDelayResyncComplete} is devoted to the proof of \cref{thm:completeFunction} and \cref{sec:applications} illustrates some consequences of \cref{thm:completeFunction}.
As the delay between streaming string transducers only depends on their induced transductions with origin, we are able to state a more general result.
\cref{cor:general} uses the fact that deterministic streaming string transducers, deterministic two-way transducers and MSO transducers are equally expressive -- they characterize the so-called \emph{regular functions} -- and every regular function given in one formalism can be translated into every other one without changing its induced transduction with origins \cite{DBLP:books/daglib/0030804,DBLP:journals/tocl/EngelfrietH01,DBLP:conf/icalp/AlurD11}\sarah{refs correct?}.

\begin{corollary}\label{cor:general}
  Given deterministic two-way transducers resp.\ MSO transducers $T_1$ and $T_2$.
  Let $R_1$ and $R_2$ denote their induced transductions with origin.
  There exist computable integers $k,\ell$ such that 
  $T_1 \equiv T_2$ iff $R_1 \equiv_{k,\ell} R_2$.
\end{corollary}


We focus on the second main aspect of our delay measure, namely, regularity.
Meaning that for every $k, \ell \in \N$, we would like to construct a finite automaton that accepts (suitable representations of) pairs $(w \otimes o_1,w \otimes o_2)$ with $\delay_{\ell}(w \otimes o_1,w \otimes o_2) \leq k$.
As finite automata enjoy good closure properties, this yields a useful tool to solve for instance decision problems up to fixed delay, cf.\ \cref{sec:applications}.
As mentioned in the paragraph before \cref{cor:general}, our delay measure is applicable to transductions with origins and is complete for several transducer models.
However, we need to pick some way to represent such transductions to show regularity.
Hence, we prove our second main result for streaming string transducers.

\begin{theorem}[Regularity]\label{thm:trueDelayResync}
  Let $\subs \subseteq \sub{\var}{\alp}$ be a \emph{finite}
subset of $\sub{\var}{\alp}$, let $k \geq 0$ and $\ell > 0$.
  The following set is a regular language:
  \[
    \mathbb{D}_{k,\ell, \subs} =
     \{\lambda \otimes \mu \in (\subs \times \subs)^* \st \delay_{\ell}(\oriout{\lambda},\oriout{\mu}) \leq k \text{ and }|\lambda| = |\mu|\}.
\]
\end{theorem}
Note that $\lambda \otimes \mu \in \mathbb{D}_{k,\ell,\subs}$ implies $\outp{\lambda} = \outp{\mu}$ because $\delay_\ell$ is defined only for such substitution sequences.
We write $\mathbb{D}_{k,\ell}$ instead of $\mathbb{D}_{k,\ell,\subs}$ when $\subs$ is clear from the context.
We prove \cref{thm:trueDelayResync} in \cref{subsec:trueDelayResync} and show some applications of this result in \cref{sec:applications}.

%% file: Sections/completness.tex
\subsection{Completeness of the delay notion}\label{subsec:trueDelayResyncComplete}

We now prove the completeness result for SST{s}, as stated in \cref{thm:completeFunction}.
%
To this end,
we show that whenever the delay between
two sequences of substitutions is sufficiently large
we can pump well-chosen factors to obtain
two sequences of substitutions producing outputs that are distinct.
Formally, given a sequence of substitutions $\lambda \in \mathcal{S}^*$
and $1 \leq s < t < |\lambda|$, we denote by
$\pump{\lambda}{(s,t]}$ the sequence of substitutions
$\lambda[1,t] \lambda(s,t] \lambda(t,|\lambda|)$
obtained by pumping the interval $(s,t]$,
and we prove the following:

\begin{lemma}\label{lem:technical}
Let $\mathcal{S}$ be a finite set of substitutions.
There exist computable integers $k,\ell \in \mathbbm{N}$
such that for every integer $C \in \mathbbm{N}$
and every pair $\lambda, \mu \in \mathcal{S}^*$
that satisfy $|\lambda| = |\mu|$,
$\outp{\lambda} = \outp{\mu}$,
and
$\delay_{C \ell}(\oriout{\lambda},\oriout{\mu}) > C^2 k$,
there exist $0 \leq t_1 < t_2 < \ldots < t_C < |\lambda|$
satisfying
\[
\outp{
\pump{\lambda}{(t_i,t_j]}}
\neq
\outp{
\pump{\mu}{(t_i,t_j]}}
\textup{ for every $1 \leq i < j\leq C$}.
\]
\end{lemma}
Moreover we can choose $k$ and $\ell$
exponential with respect to the number of variables of $\mathcal{S}$.
Before delving into the proof of \cref{lem:technical},
we argue that \cref{thm:completeFunction} follows as a corollary.

\begin{proof}[Proof of \cref{thm:completeFunction}.]
Let $\sst_1$ and $\sst_2$ be two SST{s}
with set of states $Q_1$ and~$Q_2$. By symmetry, it is sufficient to
show that $\rel{\sst_1} \subseteq \rel{\sst_2}$ iff
$R_{\origin{}}(\sst_1) \subseteq_{C^2k,C\ell} R_{\origin{}}(\sst_2)$, where $C = |Q_1| \cdot |Q_2| + 1$,
and $k$, $\ell$ are as in the statement of \cref{lem:technical}
with respect to the union $\mathcal{S} = \mathcal{S}_1 \cup \mathcal{S}_2$
of the substitutions used by $\sst_1$ and $\sst_2$. The right to left
direction of the 'iff' is immediate, as inclusion up to bounded delay
is stronger than inclusion.


We now apply \cref{lem:technical}
to prove the converse direction.
Suppose that
$R_{\origin{}}(\sst_1) \not\subseteq_{C^2k,C\ell} R_{\origin{}}(\sst_2)$.
Then there exists a pair with origins
$(u,w \otimes o_1)\in R_{\origin{}}(\sst_1)$
such that all the pairs with origin 
$(u,w \otimes o_2)\in R_{\origin{}}(\sst_2)$
with matching input and output satisfy
$\delay_{C\ell}(w \otimes o_1,w \otimes o_2) > C^2k$.
Two possible cases arise:
either there is no pair of the form
$(u,w \otimes o_2)$ in $R_{\origin{}}(\sst_2)$,
or there exists such a pair $(u,w \otimes o_2) \in R_{\origin{}}(\sst_2)$,
and it satisfies $\delay_{C\ell}(w \otimes o_1,w \otimes o_2) > C^2k$.
In the former case, we immediately get that
$R(\sst_1) \not\subseteq R(\sst_2)$,
since $(u,w)$ is in $R(\sst_1)$ but not in $R(\sst_2)$.
In the latter case, we get that there exists
a run $\rho_1$ of $\sst_1$
and a run $\rho_2$ of $\sst_2$
over the same input $u$ that
both produce the same output $w$,
but with very different origins functions:
$\delay_{C\ell}(\oriout{\kappa_1(\rho_1)},\oriout{\kappa_2(\rho_2)})
> C^2k$.
Then \cref{lem:technical} implies that we can find
$C$ intermediate points in $\kappa_1(\rho_1)$ and $\kappa_2(\rho_2)$
such that iterating the segment between any two points
yields sequences of substitutions producing distinct outputs.
As $C = |Q_1| \cdot |Q_2| + 1$,
two of these points mark a loop in both $\rho_1$ and $\rho_2$,
and pumping these loops creates runs of $\sst_1$ and $\sst_2$
with the same input but different outputs.
Since $\sst_2$ is deterministic, and therefore cannot map
an input word to two distinct output words,
this implies that $R(\sst_1)$ is not included in $R(\sst_2)$.
\end{proof}

\begin{figure}[t]
\centering
\begin{tikzpicture}[baseline,thick,scale=0.9,
  every path/.style={shorten <=0cm,shorten >=0cm}]
\pgfdeclarelayer{bg}    
\pgfsetlayers{bg,main}
\def\H{0.3}
\def\W{0.09}
\def\Hy{0.7}
\def\L{40}
\def\j{28}
\def\l{5}
\def\s{15}
\def\z{-0.75}
\def\sh{7.2}


\node[anchor=east] at (-0.1,0.4*\H){\small $\outp{\mu}$ :\strut};
\node[anchor=east] at (-0.65,0.4*\H-0.5*\z){\small $=$\strut};

\draw[fill=white!75!Yellow]
  (\j*\W - 2*\l*\W,0) rectangle (\j*\W,\H);
  
\draw[fill=white!75!Blue]
  (\j*\W,0) rectangle (\j*\W + \l*\W,\H);

\draw[thick] (0,0) -- (\L*\W,0) -- (\L*\W,\H) -- (0,\H) -- (0,0);

\node[anchor=east] at (-0.1,0.4*\H-\z){\small $\outp{\lambda}$ :\strut};

\draw[fill=white!75!Yellow]
  (\j*\W - 2*\l*\W,0-\z) rectangle (\j*\W,\H-\z);
  
\draw[fill=white!75!Blue]
  (\j*\W,0-\z) rectangle (\j*\W + \l*\W,\H-\z);

\draw[thick] (0,0-\z) -- (\L*\W,0-\z) -- (\L*\W,\H-\z) -- (0,\H-\z) -- (0,0-\z);

\draw[thick] (\j*\W,0) -- (\j*\W,\H);
\draw[thick] (\j*\W,0-\z) -- (\j*\W,\H-\z);
\draw[densely dotted] (\j*\W,-\z) -- (\j*\W, 0);
\node at (\j*\W, - 0.8*\Hy-0.25*\z){\small $j$\strut};

\draw[thick] (\j*\W - 2*\l*\W,0) -- (\j*\W - 2*\l*\W,\H);
\draw[thick] (\j*\W - 2*\l*\W,0-\z) -- (\j*\W - 2*\l*\W,\H-\z);
\draw[densely dotted] (\j*\W - 2*\l*\W,-\z) -- (\j*\W - 2*\l*\W, 0);
\node at (\j*\W - 2*\l*\W, - 0.8*\Hy-0.25*\z){\small $j_1$\strut};

\draw[thick] (\j*\W + \l*\W,0) -- (\j*\W + \l*\W,\H);
\draw[thick] (\j*\W + \l*\W,0-\z) -- (\j*\W + \l*\W,\H-\z);
\draw[densely dotted] (\j*\W + \l*\W,-\z) -- (\j*\W + \l*\W, 0);
\node at (\j*\W + \l*\W, - 0.8*\Hy-0.25*\z){\small $j_2$\strut};

\node[anchor=east] at (-0.1 + \sh,0.4*\H){\small $\outp{\pump{\mu}{(s,s']}}$ :\strut};
\node[anchor=east] at (-1.3 + \sh,0.4*\H-0.5*\z){\small $\neq$\strut};

\draw[fill=white!75!Yellow]
  (\j*\W*1.275 - 2*\l*\W + \sh,0) rectangle (\j*\W*1.275 + \sh,\H);
  
\draw[fill=white!75!Blue]
  (\j*\W*1.275 + \sh,0) rectangle (\j*\W*1.275 + \l*\W + \sh,\H);

\draw[thick] (0 + \sh,0) -- (\L*\W*1.25 + \sh,0) -- (\L*\W*1.25 + \sh,\H) -- (0 + \sh,\H) -- (0 + \sh,0);

\node[anchor=east] at (-0.1 + \sh,0.4*\H-\z){\small $\outp{\pump{\lambda}{(s,s']}}$ :\strut};

\draw[fill=white!75!Yellow]
  (\j*\W*1.2 - 2*\l*\W + \sh,0-\z) rectangle (\j*\W*1.2 + \sh,\H-\z);
  
\draw[fill=white!75!Blue]
  (\j*\W*1.2 + \sh,0-\z) rectangle (\j*\W*1.2 + \l*\W + \sh,\H-\z);

\draw[thick] (0 + \sh,0-\z) -- (\L*\W*1.25 + \sh,0-\z) -- (\L*\W*1.25 + \sh,\H-\z) -- (0 + \sh,\H-\z) -- (0 + \sh,0-\z);

\draw[thick] (\j*\W*1.275 + \sh,0) -- (\j*\W*1.275 + \sh,\H);
\draw[thick] (\j*\W*1.2  + \sh,0-\z) -- (\j*\W *1.2 + \sh,\H-\z);
\draw[densely dotted] (\j*\W*1.275 + \sh,-\z+\H) -- (\j*\W*1.275 + \sh,0);
\node at (\j*\W*1.275 + \sh, - 0.8*\Hy-0.25*\z){\small $j{+}z$\strut};
\draw[densely dotted] (\j*\W*1.2 + \sh,-\z+\H) -- (\j*\W*1.2 + \sh,0);
\node at (\j*\W*1.2 + \sh, 0.6*\Hy-1.25*\z){\small $j{+}y$\strut};

\draw[thick] (\j*\W*1.275 - 2*\l*\W + \sh,0) -- (\j*\W*1.275 - 2*\l*\W + \sh,\H);
\draw[thick] (\j*\W*1.2 - 2*\l*\W + \sh,0-\z) -- (\j*\W *1.2 - 2*\l*\W + \sh,\H-\z);

\draw[thick] (\j*\W*1.275 + \l*\W + \sh,0) -- (\j*\W*1.275 + \l*\W + \sh,\H);
\draw[thick] (\j*\W*1.2 + \l*\W + \sh,0-\z) -- (\j*\W *1.2 + \l*\W + \sh,\H-\z);

\end{tikzpicture}
  \caption{Illustration of the main idea used in the proof of \cref{lem:technical}.}\label{fig:illus}
\end{figure}
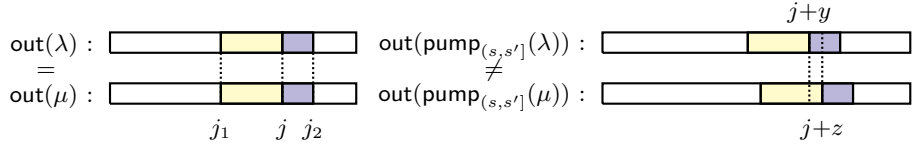

\begin{proof}[Proof overview of \cref{lem:technical}]
We consider two substitution sequences
$\lambda, \mu \in \mathcal{S}^*$ that have the same length,
produce the same output,
and satisfy 
$\delay_{C \ell}(\oriout{\lambda},\oriout{\mu}) > C^2 k$
for some $C,k,\ell \in \mathbbm{N}$.
This implies the existence of a cut position
$j \in \cut_{C \ell}(\outp{\lambda})$
and a point $t \in [0,|\lambda|)$
for which the
weight difference $\diff_{j,t}(\oriout{\lambda},\oriout{\mu})$
is equal to $C^2 k$.

Our proof is based on the following fact,
illustrated by \cref{fig:illus}:
If we choose $k$ such that
$\diff_{j,t}(\oriout{\lambda},\oriout{\mu})$ is sufficiently large,
there are intervals $(s,s'] \subset [1,|\lambda|)$
that, once pumped,
add distinct amount of output before position $j$,
thus creating a misalignment between two copies
of the letter $\outp{\lambda}[j]$ in the output words generated by
$\pump{\lambda}{(s,s']}$
and
$\pump{\mu}{(s,s']}$.
Moreover,
we show that carefully identifying patterns
occurring along the two substitution sequences
also allows us to ensure that some neighborhood
$\outp{\lambda}[\neighS,\neighT] = \outp{\mu}[\neighS,\neighT]$
of the position $j$
is preserved in both
$\outp{\pump{\lambda}{(s,s']}}$
and $\outp{\pump{\mu}{(s,s']}}$,
in a way that these two copies overlap, but not perfectly
(this is the most complex part of the proof).
At this point, we use the fact that
$j$ is a cut position of $\outp{\lambda}$:
since $j$ marks the position where the period changes,
and this position is not aligned properly in
$\outp{\pump{\lambda}{(s,s']}}$
and $\outp{\pump{\mu}{(s,s']}}$,
we can derive the existence of a mismatch,
proving that $\outp{\pump{\lambda}{(s,s']}} \neq \outp{\pump{\mu}{(s,s']}}$.

Then, all that remains to do is to combine a few counting arguments
to show how to choose $k$ and $\ell$ sufficiently large
with respect to the parameters of $\mathcal{S}$
(and, crucially, independently of $C$)
so that the fact that
$\diff_{j,t}(\oriout{\lambda},\oriout{\mu}) > C^2 k$
implies the existence of $C$ consecutive
intervals $(s,s'] \subset [1,|\lambda|)$,
which, as described earlier,
satisfy
$\outp{\pump{\lambda}{(s,s']}} \neq \outp{\pump{\mu}{(s,s']}}$,
which concludes the proof of \cref{lem:technical}.
\end{proof}

%% file: Sections/regularity.tex
\subsection{Regularity of the delay notion}\label{subsec:trueDelayResync}

Our goal is to prove that the delay notion is
regular, as stated in \cref{thm:trueDelayResync}. All over this section, $k$ and $\ell$ are non-negative integers, and $\subs$ is a finite set of substitutions over
a finite set of variables $\var$ and an alphabet $\alpo$.
\cref{lem:carac} characterizes
the pairs of substitution sequences whose outputs end with a unique endmarker symbol $\dashv$ and that are \emph{not} in
$\mathbb{D}_{k,\ell,\subs}$, using properties which are independently shown to be regular.
We first start with the characterization, then give an overview on how to show its regularity.

\myparagraph{A characterization.} 
Note that a pair $(\lambda,\mu)$ of two substitution sequences of the same length is not in $\mathbb{D}_{k,\ell,\subs}$ if either $\outp{\lambda} \not= \outp{\mu}$ or there is a cut witnessing a delay greater than $k$. Unfortunately, the first condition $\outp{\lambda} \not= \outp{\mu}$ is not regular. So in order to characterize the complement of $\mathbb{D}_{k,\ell,\subs}$ by regular properties, we somehow have to mix conditions on differences in the output and on positions witnessing a big delay. For the corresponding formal statement in \cref{lem:carac} we need the following definition.

Given a sequence of substitutions $\lambda$ and $i\geq 0$, let
$\nextcut_\ell(i,\outp{\lambda})$ be the smallest output position $j$
such that $j>i$ and $j\in \cut_\ell(\outp{\lambda})$, if it exists. 
Note that, as the last output position is a cut,
such a position $j$ always exists unless $i$ is the last output position.
Formally, 
$\nextcut_\ell(i,\outp{\lambda})$ denotes the set
$\min\left(\cut_\ell(\outp{\lambda}) \cap \{ j \mid j > i\}\right)$.
The result is either a singleton or the empty set.
In the former case, we write $\nextcut_\ell(i,\outp{\lambda}) = j$ instead of $\nextcut_\ell(i,\outp{\lambda}) = \{j\}$.

\begin{restatable}{lemma}{lemmaCarac}\label{lem:carac}
Let $\dashv\in\alpo$ and $\lambda$, $\mu\in\subs^*$ be sequences with $\outp{\lambda},\outp{\mu} \in (\alpo \setminus \{\dashv\})^*\!\!\dashv$ and $|\lambda| = |\mu|$.
Then $\lambda \otimes \mu\not\in \mathbb{D}_{k,\ell,\subs}$ iff 
there exists $i \in \left(\cut_\ell(\outp{\lambda}) \cap \cut_\ell(\outp{\mu})\right) \cup \{0\}$ such that $\maxdiff_i(\oriout{\lambda},\oriout{\mu}) \leq k$, and
one of the following holds:
\begin{enumerate}
  \item{\label{rat:i1}} Both $j_1 = \nextcut_\ell(i,\outp{\lambda})$ and $j_2 = \nextcut_\ell(i,\outp{\mu})$ exist, and either $j_1 \neq j_2$ or $\maxdiff_{j_1,j_2}(\oriout{\lambda},\oriout{\mu}) >~k$;

  \item{\label{rat:i2}} $\outp{\lambda}[i+b] \neq \outp{\mu}[i+b]$ for some $b\in [0,\ell^2]$.
\end{enumerate}
\end{restatable}

\begin{proof}[Proof sketch.]
  Assume that $\lambda,\mu$ satisfy \cref{rat:i1} or \ref{rat:i2} from the above statement.
  If $\outp{\lambda} \not= \outp{\mu}$, then clearly $\lambda \otimes \mu\not\in \mathbb{D}_{k,\ell,\subs}$. So assume that $\outp{\lambda} = \outp{\mu}$. Then $\cut_\ell(\outp{\lambda}) = \cut_\ell(\outp{\mu})$. Hence \cref{rat:i1} is satisfied with $j_1 = j_2 =: j$ and $\maxdiff_{j}(\oriout{\lambda},\oriout{\mu}) >~k$.
  Since $j$ is a cut, we obtain that $\delay_{\ell}(\oriout{\lambda},\oriout{\mu}) > k$ and thus $\lambda \otimes \mu\not\in \mathbb{D}_{k,\ell,\subs}$.

  Conversely, let $\lambda \otimes \mu \not\in \mathbb{D}_{k,\ell,\subs}$.
  First assume $\outp{\lambda} = \outp{\mu}$.
  Since $\lambda \otimes \mu \notin \mathbb{D}_{k,\ell,\subs}$, there exists some $j\in \cut_\ell(\outp{\lambda}) \cap \cut_\ell(\outp{\mu})$ such that $\maxdiff_j(\oriout{\lambda},\oriout{\mu}) >~k$,
  and we can satisfy \cref{rat:i1}.
  Second, we assume $\outp{\lambda} \neq \outp{\mu}$.
  Let $m$ be the position of the earliest mismatch (that is, $\outp{\lambda}[m] \neq \outp{\mu}[m]$), and $i$ be the nearest common cut to the left of the mismatch.
  If $\maxdiff_i(\oriout{\lambda},\oriout{\mu}) >~k$, we have a common cut with a too large difference before a mismatch occurs.
  We can treat this situation as if $\outp{\lambda} = \outp{\mu}$ and satisfy \cref{rat:i1} as before.
  If $\maxdiff_i(\oriout{\lambda},\oriout{\mu}) \leq k$, we show that either the mismatch is close to $i$ (that is, $m \leq i + \ell^2$) and \cref{rat:i2} is satisfied, or the mismatch causes $j_1 = \nextcut_\ell(i,\outp{\lambda})$ and $j_2 = \nextcut_\ell(i,\outp{\mu})$ to be different and \cref{rat:i1} is satisfied for $j_1 \neq j_2$.
\end{proof}

\myparagraph{Proof overview for the regularity of the delay notion.}
Let us denote by $\cdelay{\subs}$ the set of words
of the form $\lambda\otimes \mu$ such that $\lambda,\mu\in \subs^*$ 
satisfy the properties of the characterization given in \cref{lem:carac}. 
The main technical part is to show that $\cdelay{\subs}$ is
regular. Then regularity of $\mathbb{D}_{k,\ell,\subs}$ follows by
complementation and end-marker removal, which preserve regularity.

First, note that the definition of $\cdelay{\subs}$ is existential in
nature: it asks for the existence of positions in $\outp{\lambda}$
and $\outp{\mu}$ satisfying some properties. A classical way of
dealing with positions quantified existentially in automata theory is
to mark some positions in the input by using an extended alphabet,
construct an automaton over the extended alphabet, and then project
this automaton over the original alphabet. Here, the positions needed
in $\cdelay{\subs}$ are positions of $\outp{\lambda}$ and $\outp{\mu}$,
while the automata we want to construct read $\lambda$ and
$\mu$ as input. So, instead of marking input positions, we rather mark
positions in right-hand sides of updates occurring in the substitutions of
$\lambda$ and $\mu$. Let us make this more precise. First, for
$n\geq 0$, words $u$
over the alphabet $\Sigma$ are extended into \emph{$n$-marked words},
i.e., words over the alphabet $\Sigma\times 2^{\{1,\dots,n\}}$, such that
the additional information in $\{1,\dots,n\}$ precisely corresponds to
an $n$-tuple of positions $\overline{x}$ of $u$ (position $x_i$ is marked with label $i$ for all
$i\in\{1,\dots,n\}$, and we consider sets because the
same position can correspond to different components of
$\overline{x}$). By extension, we also define an operation $\triangleright$ which marks any substitution sequence $\lambda\in \sub{\var}{\Sigma}^*$ by a tuple $\overline{x}$ of positions of $\outp{\lambda}$, resulting in a substitution sequence 
$(\lambda\triangleright \overline{x})\in \sub{\var}{\Sigma\times 2^{\{1,\dots,n\}}}^*$ such that $\outp{\lambda\triangleright \overline{x}}=\outp{\lambda}\triangleright \overline{x}$.

We show that the set of substitution sequences $\lambda\triangleright i$ satisfying
$i\in\cut_\ell(\outp{\lambda})$ is regular, and similarly for
$\nextcut_\ell$. To do so, we
prove that the set  $\{ u\triangleright i\mid
i\in\cut_\ell(u)\}$ is regular (and
similarly for $\nextcut_\ell$) and then transfer this result to marked
substitution sequences, as regular languages are preserved under
inverse of SSTs.

Then, we show regularity results for predicates of the form
$\maxdiff_i(\oriout{\lambda},\oriout{\mu})\leq~k$ and
$\maxdiff_{j_1,j_2}(\oriout{\lambda},\oriout{\mu})>~k$. In the end, all parts of the property of the characterization of \cref{lem:carac} are shown to be regular, so that the
whole property can be checked by a synchronized product of
automata. Perhaps the most interesting part is how to show that the
predicate $\maxdiff_i(\oriout{\lambda},\oriout{\mu})\leq k$ is regular. More precisely, it is shown that the set of $(\lambda\triangleright i_1)\otimes
(\mu\triangleright i_2)$ such that $i_1 = i_2 = i$ and
$\maxdiff_i(\oriout{\lambda},\oriout{\mu})$ is smaller than $k$ (which is a given constant), is
regular. Let us intuitively explain why. In general, checking whether two
marked positions $i_1$ (in $\outp{\lambda}$) and $i_2$ (in $\outp{\mu}$) are equal
cannot be done in a regular way (recall that the automaton reads $\lambda \otimes \mu$ and not their outputs). 
However, if additionally, one has to
check that $\maxdiff_{i_1,i_2}(\oriout{\lambda},\oriout{\mu})$ is smaller than $k$, it
is actually regular. To do so, a finite automaton needs to monitor the
difference in the outputs produced in $\lambda$ and $\mu$ before 
positions $i_1$ and $i_2$ resp., and check that it is bounded by $k$
(otherwise it rejects). The difference must eventually reach $0$ when
the whole inputs $\lambda$ and $\mu$ have been read, to ensure
$i_1=i_2$.

We also prove that once a position $i$ in
$\outp{\lambda}$ is marked, then the next cut $j_1 = \nextcut_\ell(i,\outp{\lambda})$ can be identified in a regular way by an automaton. Similarly, given a constant $d$, the output position $i+d$ can also be identified in a regular way. This allows us to check the properties of \cref{rat:i1} and \cref{rat:i2} respectively of \cref{lem:carac}.
This concludes the overview of the proof of
\cref{thm:trueDelayResync}. 
A complexity analysis yields that the set $\mathbb{D}_{k,\ell,\subs}$
is recognizable by a DFA with a number of states doubly exponential in $\ell^3$ and in $|\var|$, and singly exponential in $k$.

%% file: Sections/applications.tex
\section{Applications of delay completeness and regularity}
\label{sec:applications}


\myparagraph{Decision problems up to fixed delay.}
Instead of comparing the transductions defined by non-deterministic SSTs
for inclusion or equivalence, which are undecidable problems
already for finite (variable-free) transducers, we show here that the
strengthening of those problems up to fixed delay, $\subseteq_{k,l}$ and $\equiv_{k,l}$, are decidable. 
The key to decide inclusion and equivalence up to fixed delay is the regularity of the delay notion
as stated in \cref{thm:trueDelayResync}. 
Hence, those problems can be reduced to classical inclusion and equivalence problems of regular languages. 


\begin{restatable}{theorem}{thmComplexityInclusion}\label{thm:complexityInclusion}
    Given integers $k,\ell$, the $(k,\ell)$-inclusion problem for
    non-deterministic SSTs is decidable. It is
    \PSPACE-complete if $k,\ell$ are constants and the number of variables $|\var|$
    is a constant. The same results hold for the $(k,\ell)$-equivalence problem.
\end{restatable}

\begin{proof}[Proof sketch.]
  We sketch the result for the inclusion problem. 
  For a non-deterministic SST $T$ over input alphabet $\Sigma$, we denote by $L(T)$ its language,
defined as the set of words of the form $u\otimes \tau(\rho)$, where
$u\in \Sigma^*$, $\rho$ is an accepting run of $T$ over $u$, and
$\tau(\rho)$ is the sequence of substitutions occurring on $\rho$. 
Let $T_1$ and $T_2$ be two non-deterministic SSTs over two finite sets of variables
$\var_1$ and $\var_2$ respectively, both with output variable $\varO$. Let $\subs_1$ (resp. $\subs_2$) be the finite set of
substitutions occurring in $T_1$ (resp. $T_2$). Let $\subs =
\subs_1\cup \subs_2$ and $\var = \var_1\cup \var_2$. Let
$\ell,k\in\N$.

  We let
    $\mathbb{D}^{in}_{k,\ell,\subs} = \{ u \otimes \lambda \otimes
    \mu \mid u\in\Sigma^*\wedge \lambda \otimes \mu \in \mathbb{D}_{k,\ell,\subs}\}$.  
    The automaton recognizing $\mathbb{D}_{k,\ell,\subs}$ from the proof of
    \cref{thm:trueDelayResync} can be easily extended into an automaton which recognizes
    $\mathbb{D}^{in}_{k,\ell,\subs}$.
    Now, observe that $T_1$ is $(k,\ell)$-included in $T_2$ iff $L(T_1)\subseteq \mathbb{D}^{in}_{k,\ell,\subs}(L(T_2))$, where $\mathbb{D}^{in}_{k,\ell,\subs}(L(T_2))$ denotes the set $\{ u \otimes \lambda \mid \exists u \otimes \mu \in L(T_2)\colon u \otimes \lambda \otimes \mu \in \mathbb{D}^{in}_{k,\ell,\subs}\}$.
\end{proof}

We turn to the variable minimization problem which is open for (deterministic) \SST{s} and undecidable for non-deterministic \SST{s}.
We prove decidability for variable minimization up to fixed delay in both cases.

\begin{restatable}{theorem}{thmDelayVarMin}\label{thm:delayVarMin}
    Given integers $k,\ell,m $ and a non-deterministic SST $T$, it is decidable whether $T \equiv_{k,\ell} T'$ for some (non-deterministic resp.\ deterministic) SST $T'$ that uses at most $m$ variables.
\end{restatable}

\begin{proof}[Proof sketch.]
  First, we sketch the result for non-deterministic SST. 
  Since we are looking for some non-deterministic SST $T'$ with $m$ variables such that $T \equiv_{k,\ell} T'$,
  we need to consider SSTs that produce at most $r := 2k + p$ letters (where $p$ is the maximal number of letters produced by $T$ in one step)
  per computation step in order to not violate the delay bound.
  The reasoning is that the difference between the output of the computations can be at most $k$ letters, then in the next computation step, the computation that was $k$ letters ahead may produce $p$ letters, the other computation must recover the difference by producing at least $p$ letters and at most $2k+p$ letters to keep the difference at most $k$.
  Thus, let $\subs = \subs_{T}\cup \subs_{r,m}$ and $\var = \var_{T}\cup \var_{m}$, where $\subs_{T}$ (resp. $\var_{T}$) are the substitutions (resp.\ variables) occurring in $T$,
  $\var_{m} = \{X_1,\cdots,X_m\}$,
  and $\subs_{r,m}$ are substitutions over $\var_{m}$ producing at most $r$ letters.
  
  As as in the proof sketch of \cref{thm:complexityInclusion}, $L(T)$ is the set of words of the form $u\otimes \tau(\rho)$, where $u\in \Sigma^*$, $\rho$ is an accepting run of $T$ over $u$, and $\tau(\rho)$ is the sequence of substitutions occurring on $\rho$. 
  Furthermore, as in the proof sketch above, we let
  $\mathbb{D}^{in}_{k,\ell,\subs} = \{ u \otimes \lambda \otimes
  \mu\mid u\in\Sigma^*\wedge \lambda \otimes \mu\in \mathbb{D}_{k,\ell,\subs}\}$.  
  Now, observe that there exists some non-deterministic SST $T'$ with $m$ variables such that $T \equiv_{k,\ell} T'$ iff $L(T)\subseteq \mathbb{D}^{in}_{k,\ell,\subs}(L)$ where $L$ is the set of all words $u \otimes \mu$ such that $u\in \alp^*$ and $\mu \in \subs_{r,m}^*$.
  As in the proof sketch above, $\mathbb{D}^{in}_{k,\ell,\subs}(L)$ denotes the set $\{ u \otimes \lambda \mid \exists u \otimes \mu \in L\colon u \otimes \lambda  \otimes \mu \in \mathbb{D}^{in}_{k,\ell,\subs}\}$.

  If the goal is to have a deterministic SST, we need to adapt the above procedure.
  Let $M$ be the projection of $\mathbb{D}^{in}_{k,\ell,\subs}$ onto its first and third component.
  Hence, $M$ contains all $u \otimes \mu$ such that there is some $u \otimes \lambda \in L(T)$ with $\delay_\ell(\oriout{\lambda},\oriout{\mu}) \leq k$.
  We reduce the problem to a safety game played on a DFA for $M$.
  In alternation, one player provides an input letter, the other chooses a matching transition in the DFA.
  If the input player has provided a sequence $u \in \mathrm{dom}(R(T))$ then the play must be in an accepting state of the DFA, otherwise the output player has lost.
  We show that the output player has a winning strategy iff there exists an equivalent deterministic SST with at most $m$ variables.
\end{proof}

\myparagraph{Comparison with delay for finite transducers.}
We compare our notion of delay for streaming string transducers with the
previously existing delay notion for finite
transducers~\cite{DBLP:conf/icalp/FiliotJLW16}. A \emph{rational} SST
is a non-deterministic SST with only one variable $\varO$, and updates
all of the form $\varO := \varO u$.
In other words, a rational SST $T$ is simply a finite transducer.
Applying our delay notion to rational \SST{s}  yields that if $T_1$ is $(k,\ell)$-included in $T_2$ for two rational \SST{s} $T_1, T_2$,
then $T_1$ is $k$-included in $T_2$ for the definition according to \cite{DBLP:conf/icalp/FiliotJLW16} and vice versa.
The $k$-inclusion problem for finite transducers is \PSPACE-complete for fixed $k$ \cite{DBLP:conf/icalp/FiliotJLW16}.
Hence, the complexity obtained in \cref{thm:complexityInclusion} matches this bound.
on conceptual differences between those notions.

\myparagraph{Consequences of completeness.}
We now turn to some consequences of our completeness result (\cref{thm:completeFunction}).
Inclusion for (deterministic) \SST{s} is known to be decidable \cite{DBLP:conf/popl/AlurC11}.
It is undecidable for non-deterministic SSTs, but $(k,\ell)$-inclusion is decidable (\cref{thm:complexityInclusion}).
Although \cref{thm:completeFunction} provides a new decision procedure for the inclusion problem for SSTs its value lies in showing that our notion of delay is a sensible approach to gain a better understanding of streaming string transducers.
The following two corollaries are easy consequences of \cref{thm:completeFunction}.

\begin{corollary}\label{cor:completeMinim}
Given an \SST $T$ and an integer $m$, there exist integers $k,\ell$
such that there exists an \SST $T'$ with $m$ variables such that $T \equiv T'$ iff there exists an \SST $T''$ with $m$ variables such that $T \equiv_{k,\ell} T''$.
\end{corollary}

Note that the above corollary does not imply that $k$ and $\ell$ are \emph{computable}.
This would entail a solution for the variable minimization problem for \SST{s} which is open (and decidable for concatenation-free SSTs \cite{DBLP:conf/icalp/BaschenisGMP16}).
The next result is about rational functions, that is, functions recognizable by finite transducers.
%

\begin{corollary}\label{cor:completeRat}
Given an \SST $T$, there exist integers $k,\ell$ such that there exists a rational \SST $T'$ such that $T \equiv T'$ iff there exists a rational \SST $T''$ such that $T \equiv_{k,\ell} T''$.
\end{corollary}

It was shown that it decidable whether a deterministic two-way
transducer (which is effectively equivalent to an \SST) recognizes a rational function \cite{DBLP:conf/lics/FiliotGRS13}.
The decision procedure is effective, i.e., a finite transducer is constructed if possible.
A procedure with improved complexity was given that yields a finite transducer of doubly exponential size in \cite{DBLP:conf/lics/BaschenisGMP17}.
While computability is not implied by \cref{cor:completeRat}, note
that one could compute $k,\ell$ satisfying the statement of
\cref{cor:completeRat} using \cref{thm:completeFunction} from an
equivalent rational SST (if it exists) that has been obtained using
the decision procedure from \cite{DBLP:conf/lics/BaschenisGMP17}.

%% file: Sections/conclusion.tex

\myparagraph{Other applications.} We mention other potential applications of our delay notion that ought
to be investigated. For instance, a \emph{decomposition theorem} for SST relations is still only conjectured:
Can every finite-valued SST relation be decomposed
into a finite union of SST functions?
In other settings where the corresponding statement holds
(finite transducers~\cite{DBLP:journals/ita/Weber96},
single-variable SST~\cite{DBLP:conf/stacs/GallotMPS17}),
the main ingredients of the proof
is the regularity and completeness
of the appropriate notion of delay. So, having a good delay notion
seems necessary to obtain such a result, but solving the decomposition
theorem for SSTs does \emph{not} seem to be a low-hanging fruit of our present
study.

The notion of delay might also help to solve the \emph{variable minimization} problem:
Can we determine the minimal number of variables
needed to define a given SST function?
\cref{cor:completeMinim}
makes some progress towards a positive answer,
yet it remains to prove that
the integers $k$ and $\ell$ of the statement are computable,
which is likely a complex problem.
Another interesting research direction is to study 
how our notion of delay fares beyond SST.
For \emph{copyful} SST
(where the copyless restriction of the substitutions is dropped),
our notion of delay 
can be defined in the same manner,
but its properties are unclear:
our proofs of regularity and completeness
both crucially rely on the copyless assumption.

%% file: Appendix/discussion-delay-notion.tex
\section{Discussion on natural delay notions for SSTs}


We illustrate here two natural delay notions for SSTs and
their flaws, using the executions of \cref{fig:delay}.

\begin{itemize}
    \item
    The measure $\delay_\infty$
    computes the maximal
    difference between the sizes of
    the partial outputs occurring along two runs.
    For instance,
    since the runs
    $\rho_2$, $\rho_3$, $\rho_4$ and $\rho_5$
    all steadily produce two output symbols
    (albeit different ones)
    at each step, $\delay_\infty(\rho_i,\rho_j) = 0$
    for all $i,j \in \{2,3,4,5\}$.
    This measure is complete for SSTs,
    but is not regular:
    as witnessed by Figure \ref{fig:delay},
    a small delay does not imply any similarities
    between the partial output produced,
    thus unbounded memory might be required
    to ensure that the final outputs are equal.
    This causes the undecidability of the related decision problems.
    \item
    The measure $\delay_0$
    computes the maximal size of the \emph{symmetric difference}
    between the partial outputs occurring along two runs.
    For instance
    $\delay_0(\rho_2,\rho_4) = \delay_0(\rho_3,\rho_5) = 4$ and
    $\delay_0(\rho_2,\rho_3) = \delay_0(\rho_4,\rho_5) = 8$,
    witnessed at $t=2$.
    While the high value of
    $\delay_0(\rho_2,\rho_3)$ is desirable
    ($\rho_2$ starts by producing only $a$'s
    while $\rho_3$ produces $b$'s),
    the high value of $\delay_0(\rho_4,\rho_5)$
    is problematic:
    the partial outputs of the two runs
    are not \emph{so} different,
    as they commute in the final word.
    This behavior leads to the incompleteness of $\delay_0$
    for functional SSTs:
    for instance, over a single-letter alphabet
    the SST that copies the input
    and the SST that 
    reverses it are equivalent,
    yet their delay is unbounded.
\end{itemize}
\phantom{bla}

We refer the reader to \cref{sec:comparison} for a discussion of the differences between delay notions for finite transducers and streaming string transducers.

%% file: Appendix/comparison.tex
\section{Comparison with delay for finite transducers}
\label{sec:comparison}

\myparagraph{Finite Transducers.}
Formally, a \emph{(real-time finite state) transducer} (\NFT) is a tuple $\tra = (\alp,Q,I,\Delta,F,\kappa)$
where $\aut_\tra = (\alp,Q,I,\Delta,F)$ is an automaton,
called \emph{underlying automaton} of $\tra$,
and $\kappa\colon \Delta \rightarrow \alpo^*$ is an output function.
A run of $\tra$ is a run
\[
  \rho = q_0 a_1 q_1 a_2 q_2 \cdots a_n q_n \in Q (\alp Q)^*
\]
of $\aut_\tra$,
and the \emph{output} of $\rho$ is the word 
\[
  \kappa(\rho) = \kappa((q_0,a_1,q_1)) \kappa((q_1,a_2,q_2)) \cdots \kappa((q_{n-1},a_n,q_n)) \in \alpo^*.
\]
The \emph{transduction} $\rel{\tra}$ recognized by $\tra$ is
the set of pairs $(u,v) \in \alp^* \times \alpo^*$
such that $\tra$ has an accepting run with input $u$ and output $v$.
If $\rel{\tra}$ is a (partial) function, $\tra$ is called \emph{functional}.
If $\kappa:\Delta \rightarrow \Sigma$, then $T$ is called \emph{letter-to-letter}.

\myparagraph{Rational SSTs.}
As already said in the main body of the paper, any rational SST $T$ can be
seen as a finite transducer, that we denote by $\toNFT(T)$, and which is defined as expected.
  As mentioned before, the delay between two finite transducer computations with the same input and output is a measure of how ahead one output is compared to the other.
  More formally\footnote{In \cite{DBLP:conf/icalp/FiliotJLW16}, delay is defined between synchronizations instead of runs, for ease of presentation, we do not formally introduce this notion here.}, let $T_1$ and $T_2$ be \NFT{s}, given runs $\rho_1 = q_0a_1q_1 \cdots a_nq_n$ and $\rho_2 = p_0a_1p_1 \cdots a_np_n$ of $T_1$ resp.\ $T_2$ with the same input and output, then $\delay(\rho_1,\rho_2)$ according to \cite{DBLP:conf/icalp/FiliotJLW16} is defined as 
    \begin{equation}\label{eq:finitetransducerdelay}
        \max_{0\leq i \leq n} ||\kappa(q_0,a_1,q_1)\cdots\kappa(q_{i-1}a_i,q_i)| - |\kappa(p_0,a_1,p_1)\cdots\kappa(p_{i-1}a_i,p_i)||.
    \end{equation}
  Note that this definition of delay pays no attention to periodic factors.
  In the proof of the following proposition, we show that this indeed plays no role for \NFT{s}.
  The reason is that \NFT{s} build their output from left to right; there are no gaps in the output (as, e.g., depicted in \cref{fig:non-reg-delay}) during the computation as for \SST{s}.
  As we explained in the introduction and in \cref{sec:resync}, is does not play a role how periodic factors are built (appending or prepending a period yields the same result), cutting the output into periodic factors avoids including irrelevant intermediate gaps in the output into the delay computation.
  Since computations of \NFT{s} do not have these gaps, there is no need to cut the word into periodic factors.
  
  \begin{restatable}{proposition}{propDelayrelation}\label{prop:delayrelation}
      Let $T_1,T_2$ be two rational SSTs, and $k\in\N$. Then, $\toNFT(T_1)$ is $k$-included in $\toNFT(T_2)$ (definition of~\cite{DBLP:conf/icalp/FiliotJLW16}) iff $T_1$ is $(k,\ell)$-included in $T_2$ for any $\ell \geq 1$. 
  \end{restatable}

  \begin{proof}
   Let $\rho_1 = q_0a_1q_1 \cdots a_nq_n$ denote a run of $\toNFT(T_1)$, and $\rho_2 = p_0a_1p_1 \cdots a_np_n$ denote a run of $\toNFT(T_2)$ with the same input and output.
   Furthermore, let $\lambda$ and $\mu$ denote the corresponding sequences of substitutions of $T_1$ and $T_2$, respectively.
  
   First, we express $\delay(\rho_1,\rho_2)$ using our notations.
   For all $0 \leq i \leq n$, let $j_i$ denote 
   \[
      \max \{|\kappa(q_0,a_1,q_1)\cdots\kappa(q_{i-1}a_i,q_i)|,|\kappa(p_0,a_1,p_1)\cdots\kappa(p_{i-1}a_i,p_i)| \}.
   \]
   We express 
   \[
    ||\kappa(q_0,a_1,q_1)\cdots\kappa(q_{i-1}a_i,q_i)| - |\kappa(p_0,a_1,p_1)\cdots\kappa(p_{i-1}a_i,p_i)||
    \]
    as $\diff_{j_i,i}(\oriout{\lambda},\oriout{\mu})$.
  Note that $\diff_{j_i,i}(\lambda,\mu) = \diff_{j_k,i}(\oriout{\lambda},\oriout{\mu})$ for all $i \leq k \leq n$ by definition of $\diff$, because there is no output right of $j_i$ after the $i$th computation step, as the output is build left to right.
  So,
  \[
  \begin{array}{rl}
      \delay(\rho_1,\rho_2) & = \max_{0\leq i \leq n} ||\kappa(q_0,a_1,q_1)\cdots\kappa(q_{i-1}a_i,q_i)| - |\kappa(p_0,a_1,p_1)\cdots\kappa(p_{i-1}a_i,p_i)||\\
      & = \max_{0 \leq i \leq n} \diff_{j_i,i}(\oriout{\lambda},\oriout{\mu})\\
      & = \max_{0 \leq i \leq n} \diff_{j_n,i}(\oriout{\lambda},\oriout{\mu})\\
      & = \maxdiff_{j_n}(\oriout{\lambda},\oriout{\mu}).
  \end{array}
  \]
  Recall that by definition, we have 
  \[
    \delay_\ell(\oriout{\lambda},\oriout{\mu}) = \max_{j \in \cut_\ell(\outp{\lambda})} \maxdiff_{j}(\oriout{\lambda},\oriout{\mu}).
  \]
  Furthermore, note that $j_n \in \cut_\ell(\outp{\lambda})$, because the last output position is always a cut.

  Assume $\delay_\ell(\oriout{\lambda},\oriout{\mu}) = d$ for some $\ell \geq 1$, we show $\delay(\rho_1,\rho_2) \leq~d$.
  Since $j_n \in \cut_\ell(\outp{\lambda})$, we have $\maxdiff_{j_n}(\oriout{\lambda},\oriout{\mu}) \leq d$.
  As stated above, $\delay(\rho_1,\rho_2) = \maxdiff_{j_n}(\oriout{\lambda},\oriout{\mu})$.
  Hence, $\delay(\rho_1,\rho_2) \leq d$.
  
  For the other direction, assume $\delay(\rho_1,\rho_2) = d$.
  We show $\delay_\ell(\lambda,\mu) \leq d$ for any $\ell \geq 1$, i.e., we prove that $\maxdiff_{j}(\oriout{\lambda},\oriout{\mu}) \leq d$ for any $j \in  \cut_\ell(\outp{\lambda})$.
  For each such $j$ there is some $i < n$ such that $j_i < j \leq j_{i+1}$.
  This means at least one of $\lambda$ and $\mu$ produces output at position $j$ in the $i+1$th computation step and the other computation produces this output later (or also in the $i+1$th step).
  After the $i$th step, we have $\diff_{j_i,i}(\oriout{\lambda},\oriout{\mu}) = d_i \leq d$, and after the $i+1$th step, we have $\diff_{j_{i+1},i+1}(\oriout{\lambda},\oriout{\mu}) = d_{i+1} \leq d$.
  Since after the $i$th step the outputs have length at most $j_i < j$, we can conclude that $\min \{d_i,d_{i+1}\} \leq \diff_{j,i+1}(\oriout{\lambda},\oriout{\mu}) \leq \max \{d_i,d_{i+1}\}$.
  
  As a last step, note that $\maxdiff_j(\oriout{\lambda},\oriout{\mu})$ is the maximum of
  \[
    \max_{0 \leq i} \diff_{j_i,i}(\oriout{\lambda},\oriout{\mu}) \text{ and } \diff_{j,i+1}(\oriout{\lambda},\oriout{\mu}).
\]
  Hence, we obtain that $\maxdiff_j(\oriout{\lambda},\oriout{\mu}) \leq~d$.
  
  We have proven that $\delay_\ell(\oriout{\lambda},\oriout{\mu}) \leq~d$ for any $\ell \geq 1$.
  \end{proof}

We now turn our attention to some subtleties of the delay notion for streaming string transducers.
We start with an example.

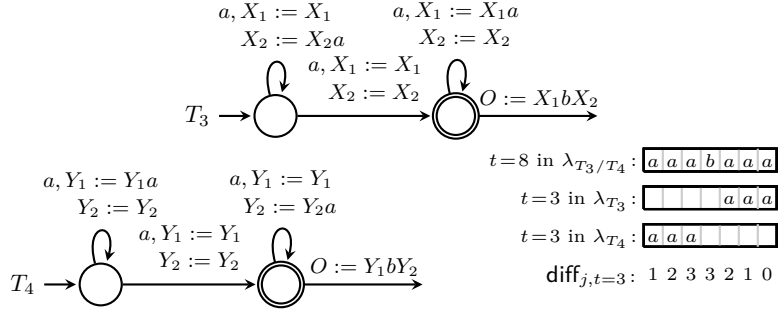
\begin{figure}[t]
    \centering
    \begin{tikzpicture}[baseline,thick,scale=1,every node/.style={scale=0.9},every loop/.style={looseness=10}]
        \tikzstyle{every state}+=[inner sep=3pt, minimum size=1.75em];
        \node[state, initial, initial text = $T_3$] (0) {};
        \node[state,right of = 0, xshift=1.5em, accepting] (1) {};

        \coordinate[right of = 1] (e);
    
        \draw[->] (1) edge[] node {\small $O:=X_1bX_2$} (e);
        \draw[->] (0) edge[loop above] node {\small $\begin{matrix} a, X_1 := X_1 \\ \phantom{a, x}X_2 := X_2a \end{matrix}$} ();
        \draw[->] (0) edge[auto] node {\small $\begin{matrix} a, X_1 := X_1\\ \phantom{a, }X_2 := X_2 \end{matrix}$} (1);
        \draw[->] (1) edge[loop above] node {\small $\begin{matrix} a, X_1 := X_1a \\ \phantom{a, }X_2 := X_2 \end{matrix}$} ();
      \end{tikzpicture}

      \begin{tikzpicture}[baseline,thick,scale=1,every node/.style={scale=0.9},every loop/.style={looseness=10}]
        \tikzstyle{every state}+=[inner sep=3pt, minimum size=1.75em];
        \node[state, initial, initial text = $T_4$] (0) {};
        \node[state,right of = 0, xshift=1.5em, accepting] (1) {};

        \coordinate[right of = 1] (e);
    
        \draw[->] (1) edge[] node {\small $O:= Y_1bY_2$} (e);
        \draw[->] (0) edge[loop above] node {\small $\begin{matrix} a, Y_1 := Y_1a \\ \phantom{a, x}Y_2 := Y_2 \end{matrix}$} ();
        \draw[->] (0) edge[auto] node {\small $\begin{matrix} a, Y_1 := Y_1\\ \phantom{a, }Y_2 := Y_2 \end{matrix}$} (1);
        \draw[->] (1) edge[loop above] node {\small $\begin{matrix} a, Y_1 := Y_1 \\ \phantom{a, }Y_2 := Y_2a \end{matrix}$} ();
      \end{tikzpicture}\nolinebreak
      \quad\quad\begin{tikzpicture}[baseline,thick,scale=1,
        every path/.style={shorten <=0cm,shorten >=0cm}]
      \pgfdeclarelayer{bg}    
      \pgfsetlayers{bg,main}
      \def\H{0.3}
      \def\W{0.25}
      \def\Hx{2.7}
      \def\Hy{0.5}

      \def\s{0.6}
      
        \node[anchor=east] at (-0.85+\Hx,0*\Hy + 0.4*\H){\small $\diff_{j,t=3}\text{\scriptsize$\colon$}$ \strut};
        \node[anchor=east] at (-0.85+\Hx,1*\Hy + 0.4*\H){\scriptsize $t\! = \! 3 \text{ in } \lambda_{T_4}\colon$ \strut};
        \node[anchor=east] at (-0.85+\Hx,2*\Hy + 0.4*\H){\scriptsize $t\! = \! 3 \text{ in } \lambda_{T_3}\colon$ \strut};
        \node[anchor=east] at (-0.85+\Hx,3*\Hy + 0.4*\H){\scriptsize $t\! = \! 8 \text{ in } \lambda_{T_3/T_4}\colon$ \strut};

      \foreach \y/\i/\c in {  0/0/1,0/1/2,0/2/3,0/3/3,0/4/2,0/5/1,0/6/0,
                              1/0/$a$,1/1/$a$,1/2/$a$,1/3/,1/4/,1/5/,1/6/,
                              2/0/,2/1/,2/2/,2/3/,2/4/$a$,2/5/$a$,2/6/$a$,
                              3/0/$a$,3/1/$a$,3/2/$a$,3/3/$b$,3/4/$a$,3/5/$a$,3/6/$a$}{

      \ifthenelse{\i>0 \AND \y>0}
      {\draw[thick] (\s*\Hx,\y*\Hy) -- (7*\W+\s*\Hx,\y*\Hy) -- (7*\W+\s*\Hx,\H+\y*\Hy) -- (\s*\Hx,\H+\y*\Hy) -- (\s*\Hx,\y*\Hy);
      \draw[black!20!white] (\i*\W+\s*\Hx,\y*\Hy) -- (\i*\W+\s*\Hx,\H+\y*\Hy);
      
      }
      {}
      \node at (\i*\W+\s*\Hx+0.5*\W,\y*\Hy + 0.4*\H){\scriptsize \c \strut};
      }      
      \end{tikzpicture}
    \caption{Two \SST{}s that realize the same relation, $\lambda_{T_3}$ resp.\ $\lambda_{T_4}$ is the sequence of substitutions of the run of $T_3$ resp.\ $T_4$ on $aaaaaaa$ where each loop is taken $3$ times, cf.\ \cref{ex:unboundeddelay}.}\label{fig:non-reg-delay}
\end{figure}

\begin{example}\label{ex:unboundeddelay}
Depicted in \cref{fig:non-reg-delay} are \SST{}s $T_3$ and $T_4$
that realize the transduction $R(T_3) = R(T_4) = \{ (a^{m_1+m_2+1}, a^{m_1}ba^{m_2}) \mid m_1,m_2 \geq 0 \}$.
Let $\rho_{T_3}$ resp.\ $\rho_{T_4}$ be the run of $T_3$ resp.\ $T_4$ on the input word $a^7$ where both loops are taken three times, let $\lambda_{T_3}$ resp. $\lambda_{T_4}$ denote the associated substitution sequence of length 8 (because we have a final substitution).
We have $\outp{\lambda_{T_3}} = \outp{\lambda_{T_4}} = aaabaaa$.
In \cref{fig:non-reg-delay}, on the right-hand side, is a visualization of the output produced by $\lambda_{T_3}$ resp. $\lambda_{T_4}$ after 3 computation steps.
The output letters are placed at the positions they will have in the final output word.
The intermediate step shows that the common output word is built quite differently.
The figure displays $\diff_{j,3}(\oriout{\lambda_{T_3}},\oriout{\lambda_{T_4}})$ for all output positions $j$.
Computing the maximal difference over all computation steps, we have, e.g, $\maxdiff_4(\oriout{\lambda_{T_3}},\oriout{\lambda_{T_4}})=3$, $\maxdiff_6(\oriout{\lambda_{T_3}},\oriout{\lambda_{T_4}})=1$ and $\maxdiff_8(\oriout{\lambda_{T_3}},\oriout{\lambda_{T_4}})=0$.
All in all, we obtain $\delay_{1}(\oriout{\lambda_{T_3}},\oriout{\lambda_{T_4}})=3$, because the $1$-cuts of $aaabaaa$ are at $3,4,8$.
\end{example}

We recall that we defined our delay notion (for streaming string transducers) with two properties in mind: Completeness (\cref{thm:completeFunction}) in order to have generality and regularity (\cref{thm:trueDelayResync}) in order to ensure good decidability properties.
The delay notion for finite transducers introduced in \cite{DBLP:conf/icalp/FiliotJLW16} also has these qualities.

We have explained in \cref{sec:resync} why we introduced the notion of cuts.
These ensure to measure delay at relevant output positions, a key ingredient to achieve completeness (while maintaining regularity which we explain below).
We have explained at the beginning of this section and shown in the proof of \cref{prop:delayrelation} that the notion of cuts is irrelevant for a good notion of delay for finite transducers.

We now want to highlight a key difference between the delay notion for SSTs and the delay notion for finite transducers that is needed to obtain regularity.
Recall the definition of delay between two finite transducer runs (\cref{eq:finitetransducerdelay}): the delay is measured between \emph{intermediate} outputs (produced on the so far processed input), then the maximum of these delays is taken.
This is in stark contrast to our delay notion: The delay is measured once between the \emph{complete} outputs.
We come back to \cref{ex:unboundeddelay} to explain why this is necessary to obtain regularity.


\begin{example}
We generalize \cref{ex:unboundeddelay}.
Given $m, n \in \N$, let $\alpha_{m,n}$, resp.\ $\beta_{m,n}$, be the sequence of substitutions of $T_3$, resp. $T_4$, over the input $a^{m+n+1}$ that takes the first loop $m$ times and the second loop $n$ times.
The output $a^{n}ba^{m}$ of $\alpha_{m,n}$ and the output $a^{m}ba^{n}$ of $\beta_{m,n}$ are equal if and only if $m = n$,
hence the set $M := \{ \alpha_{m,n} \otimes \beta_{m,n} \mid \outp{\alpha_{m,n}} = \outp{\beta_{m,n}}\}$ is not regular:
an unbounded memory is required to check that the first loop and the second loop are taken the same number of times.

Let us measure the delay between $\alpha_{m,n}$ and $\beta_{m,n}$ for all even $m,n \in \N$ using the delay notion introduced for finite transducers:
Until the last substitution, the so far produced outputs cannot be differentiated (both are $a$-blocks of the same length), there is no delay.
Finally, the last substitutions insert a $b$, cutting the $a$-blocks into two.
In order to know if the outputs are still equal, it would have been necessary to keep track of the number of times each loop is taken.

The fact that one needs to keep track of the number of times each loop is taken is reflected by our delay notion:
We obtain $\delay_{1}(\oriout{\alpha_{m,n}},\oriout{\beta_{m,n}}) = m$ for all even $m=n\in \N$, because we consider the $1$-factorization of the \emph{complete} output, namely, $a^{m}|b|a^{m}|$.
Thus, measuring the difference between outputs at the \emph{cuts} yields that $\maxdiff_{m}(\oriout{\alpha_{m,n}},\oriout{\beta_{m,n}}) = m$ for all even $m=n \in \N$.
So, it is easy to see that there is no $k \in \N$ such that the regular (according to \cref{thm:trueDelayResync}) set $\{ \alpha_{m,n} \otimes \beta_{m,n} \mid \delay_{1}(\oriout{\alpha_{m,n}},\oriout{\beta_{m,n}}) \leq k\}$ is equal to~$M$.

 \end{example}

%% file: Appendix/completeness-details.tex
    \section{Completeness of the delay notion}

    This section is devoted
    to the proof of \cref{lem:technical}.
    %
    %
    %
    We begin by setting the notation used throughout the proof.
    Let $\mathcal{S}$ be a finite substitution alphabet
    with finite alphabet $\alp$ and set of variables $\var$.
    We fix three integers $C,k,\ell \in \mathbbm{N}$
    and two sequences of substitutions
    $\lambda = \sigma_1 \sigma_2 \ldots \sigma_n \in \mathcal{S}^*$
    and $\mu = \tau_1 \tau_2 \ldots \tau_n \in \mathcal{S}^*$
    satisfying
    $|\lambda| = |\mu|$,
    $\outp{\lambda} = \outp{\mu}$,
    and 
    $\delay_{C \ell}(\oriout{\lambda},\oriout{\mu}) > C^2 k$.
    We preemptively modify $\lambda_n$ and $\mu_n$
    to add $2\ell$ padding symbols $\# \notin \alp$
    at the start and end of $\outp{\lambda}$ and $\outp{\mu}$:
    this ensures that the cuts of $\outp{\lambda}$
    are not too close to its extremities
    (and does not change the delay).
    By definition,
    $\delay_{C \ell}(\oriout{\lambda},\oriout{\mu}) > C^2 k$ implies the existence
    of positions
    $j \in \cut_{C \ell}(\outp{\lambda})$
    and
    $t \in [1,n]$
    such that
    $\diff_{j,t}(\oriout{\lambda},\oriout{\mu}) > C^2 k$.
    We fix such a pair $j$ and $t$.
    The proof revolves around a careful study of what happens
    to the infix $\outp{\lambda}[j-2C\ell,j+C\ell]$ when some interval is pumped.
    To simplify the notation we let $\neighS = j-2C\ell$ and $\neighT = j+C\ell$ (see \cref{fig:notation}).
    
    \myparagraph{Structure of the proof of \cref{lem:technical}.}
    We show that if $k$ and $\ell$ are sufficiently large,
    then $\diff_{j,t}(\oriout{\lambda},\oriout{\mu}) > C^2 k$
    implies the existence of $0 \leq t_1 < t_2 < \ldots < t_C < |\lambda|$
    such that $\outp{\pump{\lambda}{(t_i,t_j]}}$ and $\outp{\pump{\mu}{(t_i,t_j]}}$
    are distinct
    for every $1 \leq i < j \leq C$.
    This is done in two steps.
    First, we establish a property $\textsf{P}$ for intervals $(s,s'] \subset [1,|\lambda|)$ that is sufficient for
    the words $\outp{\pump{\lambda}{(s,s']}}$ and $\outp{\pump{\mu}{(s,s']}}$ to be distinct (\cref{claim:sufficent_property}).
    It can be summarized as follows:
    If both $\outp{\pump{\lambda}{(s,s']}}$ and
    $\outp{\pump{\mu}{(s,s']}}$ contain a copy of the word
    $\outp{\lambda}[\neighS, \neighT]$
    but these copies do not overlap properly,
    the fact that $j \in \cut_{C \ell}(\outp{\lambda})$
    implies the existence of a mismatch
    between $\outp{\pump{\lambda}{(s,s']}}$ and $\outp{\pump{\mu}{(s,s']}}$.
    Second, we show how to ensure property $\textsf{P}$
    by carefully regulating the quantity of output produced
    and the arrangement of the variables with respect to
    the neighborhood $[\neighS,\neighT]$ of $\outp{\lambda}$
    and $\outp{\mu}$.
    Formally, we establish three conditions
    $\textsf{C}_1$, $\textsf{C}_2$ and $\textsf{C}_3$
    sufficient to get property $\textsf{P}$
    (\cref{claim:sufficient_conditions}).
    We then conclude by showing that,
    since $\diff_{j,t}(\oriout{\lambda},\oriout{\mu}) > C^2 k$,
    choosing $k$ and $\ell$ sufficiently large
    with respect to the parameters of $\mathcal{S}$
    (and, crucially, independently of $C$)
    guarantees the existence
    of $C$ consecutive intervals satisfying
    $\textsf{C}_1$, $\textsf{C}_2$ and $\textsf{C}_3$
    (\cref{claim:guarantee_conditions}).
    Combining the three claims yields the proof of \cref{lem:technical}.
    
    \myparagraph{One sufficient property.}
    An interval $(s,s'] \subset [1,|\lambda|)$
    \emph{satisfies Property $\textsf{P}$}
    if the output words generated by
    $\pump{\lambda}{(s,s']}$
    and
    $\pump{\mu}{(s,s']}$
    both contain a copy of the word
    $\outp{\lambda}[\neighS,\neighT] = \outp{\mu}[\neighS,\neighT]$
    such that these copies overlap partially,
    with a shift of at most $C \ell/2$.
    Formally:
    
    \begin{description}
        \item[$\textsf{P}$.]
        There exist $x,y \in \mathbbm{N}$ such that:
    $
    \left\{
    \begin{array}{lll}
    \outp{\pump{\lambda}{(s,s']}}[\neighS + x,\neighT + x] = \outp{\lambda}[\neighS,\neighT];\\
    \outp{\pump{\mu}{(s,s']}}[\neighS+y,\neighT + y] = \outp{\mu}[\neighS,\neighT];\\
    0 < |x-y| \leq C\ell.
    \end{array}
    \right.$
    \end{description}
    
    Since $j \in \cut_{C \ell}(\outp{\lambda})$
    we can derive from Property $\textsf{P}$
    the existence of a mismatch between
    $\outp{\pump{\lambda}{(s,s']}}$ and $\outp{\pump{\mu}{(s,s']}}$.

    \begin{figure}[t]
    \begin{tikzpicture}[baseline,thick,scale=1,
      every path/.style={shorten <=0cm,shorten >=0cm}]
    \pgfdeclarelayer{bg}    
    \pgfsetlayers{bg,main}
    \def\H{0.3}
    \def\W{0.25}
    \def\Hy{1}
    \def\L{50}
    \def\j{38}
    \def\l{5}
    \def\s{15}
    

    \node[anchor=east] at (-0.1,0.4*\H){\small $\outp{\lambda}$ :\strut};

    \draw[fill=white!75!Yellow]
      (\s*\W,0) rectangle (\j*\W,\H);
    \draw[decoration={brace,mirror,raise=5pt},decorate,thick]
      (\j*\W,\H) -- node[above=6pt] {\scriptsize $p$-periodic} (\s*\W,\H);

    \draw[fill=white!75!Blue]
      (\j*\W - 2*\l*\W,0) rectangle (\j*\W + \l*\W,\H);
    \draw[decoration={brace,mirror,raise=5pt},decorate,thick]
      (\j*\W - 2*\l*\W,0) -- node[below=6pt] {\scriptsize $q$-periodic} (\j*\W + \l*\W,0);  
      
    \draw[fill=white!75!ForestGreen]
      (\j*\W - 2*\l*\W,0) rectangle (\j*\W,\H);
    
    \draw[thick] (0,0) -- (\L*\W,0) -- (\L*\W,\H) -- (0,\H) -- (0,0);

    \draw[thick] (\s*\W,0) -- (\s*\W,\H);
    \draw[dotted] (\s*\W,0) -- (\s*\W, - 0.5*\Hy);
    \node at (\s*\W, - 0.8*\Hy){\small $j'$\strut};

    \draw[thick] (\j*\W,0) -- (\j*\W,\H);
    \draw[dotted] (\j*\W,0) -- (\j*\W, - 0.5*\Hy);
    \node at (\j*\W, - 0.8*\Hy){\small $j$\strut};

    \draw[thick] (\j*\W - 2*\l*\W,0) -- (\j*\W - 2*\l*\W,\H);
    \draw[dotted] (\j*\W - 2*\l*\W,0) -- (\j*\W - 2*\l*\W, - 0.5*\Hy);
    \node at (\j*\W - 2*\l*\W, - 0.8*\Hy){\small $j_1 = j - 2C\ell$\strut};

    \draw[thick] (\j*\W + \l*\W,0) -- (\j*\W + \l*\W,\H);
    \draw[dotted] (\j*\W + \l*\W,0) -- (\j*\W + \l*\W, - 0.5*\Hy);
    \node at (\j*\W + \l*\W, - 0.8*\Hy){\small $j_2 = j + C\ell$\strut};

    \end{tikzpicture}
      \caption{The different factors of $\outp{\lambda}$ studied in the proof of \cref{claim:sufficent_property}.}\label{fig:notation}
    \end{figure}
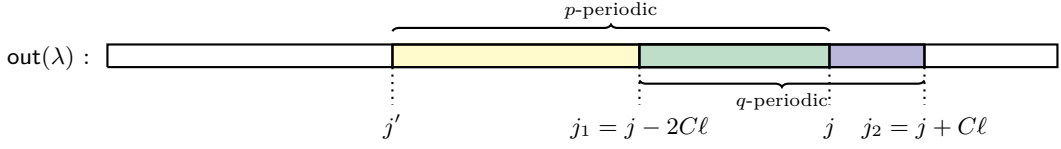

    \begin{claim}
    \label{claim:sufficent_property}
    For each interval $(s,s']$ satisfying $\textsf{P}$,
    $\outp{\pump{\lambda}{(s,s']}} \neq \outp{\pump{\mu}{(s,s']}}$.
    \end{claim}
    
    \begin{proof}
    The proof relies on the Fine-Wilf theorem~\cite{FineW65},
    a fundamental result in the study of combinatorics on words,
    that we now remind:
    For each $p >0$,
    we say that a word $w \in \alp^*$
    is \emph{$p$-periodic}
    if $w[i] = w[i+p]$ for all $1 \leq i \leq |w|-p$.
    Given $p,q > 0$ and $u,v,w \in \alp^*$,
    if $uv$ is $p$-periodic,
    $vw$ is $q$-periodic
    and $|v| \geq p+q-gcd(p,q)$,
    then $uvw$ is $gcd(p,q)$-periodic.
    
    Let us now prove the claim.
    First, remark that since $j \in \cut_{C \ell}(\outp{\lambda})$,
    by definition of $\fact_{C\ell}(\outp{\lambda})$
    there exists an integer $j' < j$
    such that $\outp{\lambda}[j',j]$
    is the longest prefix of $\outp{\lambda}[j',|\outp{\lambda}|]$
    that has a primitive root of length smaller than $C\ell$.
    In particular, this implies that $\outp{\lambda}[j',j]$ is $p$-periodic for some $p \leq C\ell$.
    Now let $(s,s'] \subset [0,n]$ be an interval satisfying $\textsf{P}$.
    To prove the claim we suppose,
    towards building a contradiction,
    that $\outp{\pump{\lambda}{(s,s']}} = \outp{\pump{\mu}{(s,s']}}$.
    Combining Property \textsf{P} with this equality 
    implies that $\outp{\lambda}[\neighS,\neighT]$
    is $q$-periodic for some $1 \leq q \leq C\ell$.
    We show that this implies that $\outp{\lambda}[j',\neighT]$ is also $q$-periodic.
    To this end, we differentiate two cases:
    \begin{enumerate}
        \item
        If $j' \geq \neighS$ then $[j',\neighT] \subseteq [\neighS,\neighT]$,
        and we directly get that
        $\outp{\lambda}[j',\neighT]$ is $q$-periodic.
        \item
        If $j' < \neighS$ then
        $[j',j] \cap [\neighS,\neighT] = [\neighS,j]$ (See \cref{fig:notation}).
        Since $j-\neighS = 2C\ell \geq p+q$,
        the Fine-Wilf theorem yields that $\outp{\lambda}[j',\neighT]$ is $gcd(p,q)$-periodic
        (hence it is also $q$-periodic).
    \end{enumerate}
    In both cases we showed that the interval $\outp{\lambda}[j',\neighT]$ is $q$-periodic.
    Since $q \leq C\ell$
    there exists $j'' \in (j,\neighT] = (j,j + C\ell]$
    such that $j''-j'$
    is a multiple of $q$.
    Then $\outp{\lambda}[j',j'']$ is 
    a prefix of $\outp{\lambda}[j',|\outp{\lambda}|]$
    longer than $\outp{\lambda}[j',j]$
    that has a primitive root of length at most $q$,
    thus smaller than $C\ell$.
    This contradicts the definition of $j$
    as the next cut after $j'$.
    \end{proof}

    \myparagraph{Three sufficient conditions.}
    We establish three conditions on intervals $(s,s'] \subset [1,|\lambda|)$
    whose conjunction implies Property \textsf{P}.
    Condition $\textsf{C}_1$ limits the growth of the delay from $s$ to $s'$
    in order to ensure that the shift between the two copies of $\outp{\lambda}[\neighS,\neighT]$ is appropriate.
    
    \begin{description}
        \item[$\textsf{C}_1$.]
        The delay between $\lambda$ and $\mu$ at position $\neighS$ increases between $s$ and $s'$,
        but \emph{not too much}:
        $\diff_{\neighS,s}(\oriout{\lambda},\oriout{\mu}) < \diff_{\neighS,s'}(\oriout{\lambda},\oriout{\mu}) < \diff_{\neighS,s}(\oriout{\lambda},\oriout{\mu}) + C\ell$;
    \end{description}
    
    \noindent
    Conditions $\textsf{C}_2$ and $\textsf{C}_3$ are more technical.
    To guarantee property \textsf{P}
    we need to ensure that $\outp{\lambda}[\neighS,\neighT] = \outp{\mu}[\neighS,\neighT]$ still appears after $(s,s']$ is pumped.
    To this end,
    we study how the production of $\outp{\lambda}[\neighS,\neighT]$ and $\outp{\mu}[\neighS,\neighT]$ evolves between $s$ and $s'$.
    We begin by introducing some notation concerning the sequences of substitutions
    $\lambda = \sigma_1 \sigma_2 \ldots \sigma_n \in \mathcal{S}^*$
    and $\mu = \tau_1 \tau_2 \ldots \tau_n \in \mathcal{S}^*$.
    We only define the notation explicitly for $\lambda$:
    it is defined analogously for $\mu$ by replacing each occurrence of $\lambda$ with $\mu$,
    and each occurrence of $\sigma$ with $\tau$.
    
    We set $\sigma_0 = \sigma_\epsilon$ and for every interval $[s,s'] \subseteq [0,n]$
    we denote by $\sigma_{[s,s']}$ the composition
    $\sigma_s \circ \sigma_{s+1} \circ \ldots \circ \sigma_{s'}$.
    Using this notation,
    we remark that the words
    $\outp{\lambda} = \sigma_{[0,s']}(\sigma_{(s',n]}(O))$
    and 
    $\outp{\pump{\lambda}{(s,s']}} = \sigma_{[0,s']}(\sigma_{(s,n]}(O))$
    are obtained by applying 
    the same substitution $\sigma_{[0,s']}$ to either $\sigma_{(s',n]}(O)$ or $\sigma_{(s,n]}(O)$.
    Since Property \textsf{P} requires that the infix $\outp{\lambda}[\neighS,\neighT]$ still appears after pumping $(s,s']$,
    we are specifically interested in the parts of $\sigma_{(s,n]}(O)$ and $\sigma_{(s',n]}(O)$ that are responsible
    for the creation of $\outp{\lambda}[\neighS,\neighT]$:
    For all $s \in [1,n]$
    we define a decomposition $\prefix_\lambda(s)\infix_\lambda(s)\suffix_\lambda(s)$ of $\sigma_{(s,n]}(O)$
    by setting $\infix_\lambda(s)$ as the \emph{minimal} infix such that
    \[
    \sigma_{[0,s]}(\prefix_\lambda(s)) = \outp{\lambda}[1,\neighS'), \ \
    \sigma_{[0,s]}(\infix_\lambda(s)) = \outp{\lambda}[\neighS',\neighT'], \ \
    \sigma_{[0,s]}(\suffix_\lambda(s)) = \outp{\lambda}(\neighT',\lambda],
    \]
    for some $\neighS' \leq \neighS$ and $\neighT' \geq \neighT$.
    In other words, $\infix_\lambda(s)$ is composed of the parts of $\outp{\lambda}[\neighS,\neighT]$ appearing in $\sigma_{(s,n]}(O)$ together with
    the variables that will contain fragments of $\outp{\lambda}[\neighS,\neighT]$
    once $\sigma_{[0,s)}$ is applied.
    We now have all the tools to define the remaining two conditions:
    
    \begin{description}
        \item[$\textsf{C}_2$.]
        The parts of
        $\sigma_{(s,n]}(O)$ and $\sigma_{(s',n]}(O)$ responsible for the creation of $\outp{\lambda}[\neighS,\neighT]$ and $\outp{\mu}[\neighS,\neighT]$
        are identical:
        $\infix_{\lambda}(s)  = \infix_{\lambda}(s')$
        and $\infix_{\mu}(s)  = \infix_{\mu}(s')$.
        \item[$\textsf{C}_3$.]
        The sequence of variables containing fragments of
        $\outp{\lambda}[1,\neighS)$ and
        $\outp{\mu}[1,\neighS)$ are the same at $s$ and $s'$:
        $\eta_{\var}(\prefix_{\lambda}(s)) = \eta_{\var}(\prefix_{\lambda}(s'))$
        and $\eta_{\var}(\prefix_{\mu}(s)) = \eta_\var({\prefix_{\mu}(s')})$,
        where $\eta_\var$ is the erasing morphism that
        preserves the symbols of $\var$ and erases
        the symbols of $\alp \cup \{ \# \}$.
    \end{description}
    We show that the conjunction of the three conditions implies Property $\textsf{P}$.

    \begin{claim}
    \label{claim:sufficient_conditions}
    Every interval 
    $(s,s']$ that satisfies
    $\textsf{C}_1$, $\textsf{C}_2$ and $\textsf{C}_3$
    also satisfies $\textsf{P}$.
    \end{claim}
    

    \begin{proof}
    Let us consider an interval $(s,s']$ that satisfies the three conditions.
    We first show that Condition $\textsf{C}_2$
    directly implies that both $\outp{\pump{\lambda}{(s,s']}}$ and $\outp{\pump{\mu}{(s,s']}}$
    contain a copy of the word
    $\outp{\lambda}[\neighS,\neighT] = \outp{\mu}[\neighS,\neighT]$,
    and then we combine $\textsf{C}_1$ and $\textsf{C}_3$ to show that the overlap between these two copies has the required size. 
    
    We focus on the sequence of substitutions $\lambda$
    by comparing $\outp{\lambda}$ and $\outp{\pump{\lambda}{(s,s']}}$.
    Since $\infix_\lambda(s) = \infix_\lambda(s')$ by Condition $\textsf{C}_2$,
    there exist $\neighS'' \leq \neighS$ and $\neighT'' \geq \neighT$ satisfying
    \[
    \begin{array}{lll}
    \outp{\lambda}
    & = &
    (\sigma_{[0,s']}
    \circ \sigma_{(s',n]})(O)\\
    & = &
    \sigma_{[0,s']}(\prefix_\lambda(s') \cdot \infix_\lambda(s') \cdot \suffix_\lambda(s'))\\
    & = &
    \sigma_{[0,s']}
    (\prefix_\lambda(s')) \cdot \outp{\lambda}[\neighS'',\neighT''] \cdot \sigma_{[0,s']}(\suffix_\lambda(s'));\\
    \outp{\pump{\lambda}{(s,s']}}
    & = &
    (\sigma_{[0,s']}
    \circ \sigma_{(s,n]})(O)\\
    & = &
    \sigma_{[0,s']}(\prefix_\lambda(s) \cdot \infix_\lambda(s) \cdot \suffix_\lambda(s))\\
    & = &
    \sigma_{[0,s']}
    (\prefix_\lambda(s)) \cdot \sigma_{[0,s']}(\infix_\lambda(s')) \cdot \sigma_{[0,s']}(\suffix_\lambda(s))\\
    & = &
    \sigma_{[0,s']}
    (\prefix_\lambda(s)) \cdot \outp{\lambda}[\neighS'',\neighT''] \cdot \sigma_{[0,s']}(\suffix_\lambda(s)).
    \end{array}
    \]
    As an immediate consequence, we have that
    \[
    \outp{\pump{\lambda}{(s,s']}}[\neighS'' + x,\neighT'' + x] = \outp{\lambda}[\neighS'',\neighT''],
    \textup{ with } x = |\sigma_{[0,s']}(\prefix_\lambda(s))|-|\sigma_{[0,s']}(\prefix_\lambda(s'))|.
    \]
    Since $\neighS'' \leq \neighS$ and $\neighT'' \geq \neighT$,
    this implies that, as required, $\outp{\pump{\lambda}{(s,s']}}$ contains a copy of $\outp{\lambda}[\neighS,\neighT]$.
    Let us now focus on the value of the shift $x$.
    First, remark that by definition of the function $\prefix_\lambda$,
    the letters of $\alp$ that appear in $\prefix_\lambda(s)$, respectively $\prefix_\lambda(s')$,
    are exactly the letters of $\outp{\lambda}$
    whose position is smaller than $\neighS$ and origin is greater than $s$, respectively $s'$, hence:
    \[
    \begin{array}{llll}
     x & = &
    |\sigma_{[0,s']}(\prefix_\lambda(s))| - |\sigma_{[0,s']}(\prefix_\lambda(s'))|\\
    & = & 
    |\eta_{\alp}(\prefix_\lambda(s))| + 
    |\sigma_{[0,s']}(\eta_{\var}(\prefix_\lambda(s)))|
    -
    |\eta_{\alp}(\prefix_\lambda(s'))| - 
    |\sigma_{[0,s']}(\eta_{\var}(\prefix_\lambda(s')))|\\
    & = & 
    (|\outp{\lambda}| - \weight_{\neighS,s}(\oriout{\lambda})) + 
    |\sigma_{[0,s']}(\eta_{\var}(\prefix_\lambda(s)))|\\
    & &
    - (|\outp{\lambda}| - \weight_{\neighS,s'}(\oriout{\lambda})) -
    |\sigma_{[0,s']}(\eta_{\var}(\prefix_\lambda(s')))|\\
    \end{array}
    \]
    \sarah[inline]{I think it should be $|\outp{\lambda}|$?}
    \isma[inline]{Good catch, I changed it}
    Moreover, Condition $\textsf{C}_3$ implies that the second and fourth terms cancel each other.
    To summarize, we showed that
    \[
    \outp{\pump{\lambda}{(s,s']}}[\neighS + x,\neighT + x] = \outp{\lambda}[\neighS,\neighT],
    \]
    $\textup{with } x = \weight_{\neighS,s'}(\oriout{\lambda}) - \weight_{\neighS,s}(\oriout{\lambda}).$
    Since Conditions $\textsf{C}_2$
    and $\textsf{C}_3$ treat $\lambda$ and $\mu$ identically,
    we get similarly that
    \[
    \outp{\pump{\mu}{(s,s']}}[\neighS + y,\neighT + y] = \outp{\mu}[\neighS,\neighT],
    \]
    $\textup{with } y = \weight_{\neighS,s'}(\oriout{\mu}) - \weight_{\neighS,s}(\oriout{\mu}).$
    To prove that $(s,s']$ 
    satisfies Property \textsf{P},
    it only remains to apply $\textsf{C}_1$:
    \[
    \begin{array}{lllll}
    |x-y| & = &
    \big|\weight_{\neighS,s'}(\oriout{\lambda}) - \weight_{\neighS,s}(\oriout{\lambda})\\
    & &
    -
    \big(
    \weight_{\neighS,s'}(\oriout{\mu}) - \weight_{\neighS,s}(\oriout{\mu})
    \big)\big|\\
    & = &
    |\diff_{\neighS,s}(\oriout{\lambda},\oriout{\mu})
    -
    \diff_{\neighS,s'}(\oriout{\lambda},\oriout{\mu})|
    \ \in \ [1,C\ell/2].
    \end{array}
    \]
    This concludes the proof.
    \end{proof}
    
    
    To conclude, we show that if $k$ and $\ell$
    are large enough with respect to the parameters of $\mathcal{S}$,
    the fact that $\diff_{j,t}(\oriout{\lambda},\oriout{\mu}) > C^2k$ implies through a combination of counting arguments
    the existence of $C$ consecutive segments
    that satisfy the three conditions.
    
    \begin{claim}
    \label{claim:guarantee_conditions}
    Let $M_1$ be the size 
    of the largest substitution occurring in $\mathcal{S}$, and let
    $M_2 = (|\var| + 1)^{|\var|+1} \cdot |\var|!$.
    If $\ell = M_1 M_2^2$ and $k > (6\ell + 3) M_1 M_2^2$,
    then there exist $0 \leq t_1 < t_2 < \ldots < t_C \leq t$
    such that $(t_i,t_{i'}]$ satisfies
    $\textsf{C}_1$, $\textsf{C}_2$ and $\textsf{C}_3$
    for all $1 \leq i < i'\leq C$.
    \end{claim}
    
    
    
    \begin{proof}
    If we consider a variable $s$ ranging from $0$ to $t$,
    then the value of $\diff_{j,s}(\oriout{\lambda},\oriout{\mu})$ increases by at most $M_1$ whenever $s$ increases by $1$.
    Since
    $\diff_{j,t}(\oriout{\lambda},\oriout{\mu}) > C^2k =
    C^2 \cdot (6\ell + 3)  M_1 M_2^2$,
    this implies that we can find
    $C'= C^2 \cdot (6\ell + 3) \cdot M_2^2$
    points $s_1 < s_2 < \ldots < s_{C'}$
    between $0$ and $t$
    along which the value of $\diff_{j,s}(\oriout{\lambda},\oriout{\mu})$ increases:
    We set $s_0 = 0$, and for every
    $1 \leq i \leq C' $
    we let $s_{i}$ be the minimal integer greater than $s_{i-1}$
    that satisfies
    $\diff_{j,s_{i}}(\oriout{\lambda},\oriout{\mu})>\diff_{j,s_{i-1}}(\oriout{\lambda},\oriout{\mu})$.
    This satisfies the following variation of condition $\textsf{C}_1$:
    
    \begin{description}
        \item[$\textsf{C}_1'$.]
        $\diff_{j,s_i}(\oriout{\lambda},\oriout{\mu}) < \diff_{j,s_{i+1}}(\oriout{\lambda},\oriout{\mu}) \leq \diff_{j,s_i}(\oriout{\lambda},\oriout{\mu}) + M_1$
        for every $0 \leq i \leq C'-1$.
    \end{description}
    
    Now, in order for $\textsf{C}_2$ to be satisfied,
    in particular we need intervals along which
    no output at the positions $[\neighS,\neighT]$
    of $\outp{\lambda}$ and $\outp{\mu}$ is produced.
    Let $\oriout{\lambda} = \outp{\lambda} \otimes o_\lambda$ and $\oriout{\mu} = \outp{\mu} \otimes o_\mu$.
    The set\sarah{I redefined U using other notation as $\origin{\lambda}$ does not exist any more.}
    \isma{All good.}
    $U = \{ o_{\lambda}[i] |
    i \in [\neighS,\neighT]\} \cup \{ o_{\mu}[i] | i \in [\neighS,\neighT]\}$
    contains at most $6C\ell+2$ elements,
    thus whenever we pick $6C\ell+3$ disjoint intervals in $[0,n]$,
    at least one of them contains no element from $U$. 
    As $C' = C^2 \cdot (6\ell + 3) \cdot M_2^2 > C \cdot (6C\ell + 3) \cdot M_2^2$,
    there exist $0 \leq x < y \leq C'$
    such that $y - x = C \cdot M_2^2$
    and for all $i \in [\neighS,\neighT]$, $\origin{\lambda}(i) \notin (s_{x},s_{y})$ and $\origin{\mu}(i) \notin (s_{x},s_{y})$.
    This property greatly restricts the possible values of $\infix_\lambda(s_i)$ when $i$ ranges from $x$ to $y$:
    if we set $\infix_\lambda(s_y) = u_0 X_1 u_1 X_2 u_2 \ldots X_m u_m$ with $X_1,\ldots,X_m \in \var$ and $u_0,\ldots,u_m \in \alp^*$,
    then for every $x \leq i \leq y$,
    since $\sigma_{(i,n]} (0)= (\sigma_{(i,y]} \circ \sigma_{(y,n]}) (O)$,
    there exists $H_1,H_2, \ldots, H_m \in \var^*$ such that $\infix_\lambda(s_i) = u_0 H_1 u_1 H_2 u_2 \ldots H_m u_m$.
    In other words $\infix_\lambda(s_i)$ is entirely determined by the tuple composed of the $m$ sequences of variables
    $H_1,H_2, \ldots, H_m \in \var^*$ that occur in it,
    and if we set on top of that $H_0 = \eta_\var(\prefix_\lambda(s_i))$,
    we get that the pair $(\eta_\var(\prefix_\lambda(s_i)),\infix_\lambda(s_i))$ is entirely determined by
    a tuple of $m+1$ sequences of $\var^*$.
    Similarly, the pair $(\eta_\var(\prefix_\mu(s_i)),\infix_\mu(s_i))$ is also determined by $m+1$ sequences of $\var^*$.
    Finally, since we are dealing with copyless substitutions,
    at most $M_2 = (|\var| + 1)^{|\var| + 1} \cdot |\var|!$ tuples of $m+1$ sequences of $\var^*$ can occur:
    there are $|\var|!$ possible ways of concatenating all the variables in $\var$,
    each such sequence has $|\var| + 1$ distinct prefixes,
    and there are (at most) $(|\var| + 1)^{|\var|}$ manners to cut each prefix in a tuple of $m+1$
    (possibly empty) parts since $m \leq |X|$ (as $m$ is the number of variables appearing in $\infix_\lambda(s_y)$).
    
    Therefore, since $y - x = C M_2^2$,
    we can find $C$ indices
    $t_1,t_2,\ldots,t_C \in \{s_{x},s_{x + 1}, \ldots, s_{y}\}$
    such that for every $1 \leq i < i' \leq C$,
    $\textsf{C}_2$ and $\textsf{C}_3$ are satisfied:
    $(\infix_{\lambda}(t_i),\eta_\var(\prefix_{\lambda}(t_i))) = (\infix_{\lambda}(t_{i'}),\eta_\var(\prefix_{\lambda}(t_{i'})))$ and
    $(\infix_{\mu}(t_i),\eta_\var(\prefix_{\lambda}(t_i))) = (\infix_{\mu}(t_{i'}),\eta_\var(\prefix_{\lambda}(t_{i'})))$.
    Finally, $\textsf{C}_1$ is obtained by 
    combining $\textsf{C}_1'$ with the facts that
    $y-x = C M_2^2$ and $\ell = M_1 M_2^2$.
    \end{proof}

%% file: Appendix/regularity-details.tex
  
\section{Regularity of the delay notion}

The following sections are devoted to the missing proofs of \cref{subsec:trueDelayResync}.

\subsection{Characterization lemma}
    
    \lemmaCarac*
    \begin{proof}
      First assume that $\lambda,\mu$ satisfy \cref{rat:i1} or \ref{rat:i2} from the statement of the lemma. We show that $(\lambda,\mu)\not\in \mathbb{D}_{k,\ell,\subs}$. If $\outp{\lambda} \not= \outp{\mu}$, then clearly $(\lambda,\mu)\not\in \mathbb{D}_{k,\ell,\subs}$. So assume that $\outp{\lambda} = \outp{\mu}$. Then $\cut_\ell(\outp{\lambda}) = \cut_\ell(\outp{\lambda})$. Hence it is not possible that \cref{rat:i2} is satisfied or \cref{rat:i1} with $j_1 \not= j_2$. The remaining case is \cref{rat:i1} with $j_1 = j_2 =: j$ and $\maxdiff_{j}(\oriout{\lambda},\oriout{\mu}) >~k$. Since $j$ is a cut, we obtain that $\delay_{\ell}(\oriout{\lambda},\oriout{\mu}) > k$ and thus $(\lambda,\mu)\not\in \mathbb{D}_{k,\ell,\subs}$.

    For the other direction, consider $(\lambda,\mu) \not\in \mathbb{D}_{k,\ell,\subs}$.
    Firstly, we assume that $\outp{\lambda} = \outp{\mu}$ which entails that $\cut_\ell(\outp{\lambda}) = \cut_\ell(\outp{\mu})$.
    Since $(\lambda,\mu) \notin \mathbb{D}_{k,\ell,\subs}$, there exists some $j \in \cut_\ell(\outp{\lambda}) \cap \cut_\ell(\outp{\mu})$ such that $\maxdiff_j(\oriout{\lambda},\oriout{\mu}) > k$.
    Let $j$ be the smallest such cut position.
    If $j$ is the first cut, let $i = 0$, otherwise let $i \in \cut_\ell(\outp{\lambda}) \cap \cut_\ell(\outp{\mu})$ such that $\nextcut_\ell(i,\outp{\lambda}) = \nextcut_\ell(i,\outp{\mu}) = j$.
    Clearly, \cref{rat:i1} is satisfied for this $i$.
    
    Secondly, we assume that $\outp{\lambda} \neq \outp{\mu}$.
    Since $\outp{\lambda},\outp{\mu} \in (\alpo{-}\dashv)^*\!\!\dashv$, there exists some $m$ such that $\outp{\lambda}[m] \neq \outp{\mu}[m]$.
    Consider the smallest such $m$.
    Let $M = \left(\left(\cut_\ell(\outp{\lambda}) \cap [1,m-1]\right) \cap \left(\cut_\ell(\outp{\mu}) \cap [1,m-1]\right)\right) \cup \{0\}$.
    Pick the position $i \in M$ closest to $m$ (the maximal position in $M$).
    If $\maxdiff_i(\oriout{\lambda},\oriout{\mu}) \leq k$ and $m \leq i + \ell^2$, we are done since \cref{rat:i2} is satisfied.
    
    Otherwise, we consider two cases.
    Either $\maxdiff_i(\lambda,\mu) > k$ or $\maxdiff_i(\oriout{\lambda},\oriout{\mu}) \leq k$ and $m > i + \ell^2$.
    
    Assume that $\maxdiff_i(\oriout{\lambda},\oriout{\mu}) > k$.
    Note that $\outp{\lambda}[1,i] = \outp{\mu}[1,i]$ since $i < m$. Thus, $\maxdiff_i(\oriout{\lambda},\oriout{\mu}) > k$ and $\maxdiff_0(\oriout{\lambda},\oriout{\mu}) = 0$ implies that there exist $i',i'' \in M$ with $0 \le i' < i'' \leq i$ and $i'' = \nextcut_\ell(i',\outp{\lambda}) = \nextcut_\ell(i',\outp{\mu})$ such that $\maxdiff_{i'}(\oriout{\lambda},\oriout{\mu}) \leq k$ and $\maxdiff_{i'',i''}(\oriout{\lambda},\oriout{\mu}) > k$.
    Hence, \cref{rat:i1} is satisfied (with $i'$ for $i$ and $i''$ for $j_1$ and $j_2$).
    
    Assume that $\maxdiff_i(\oriout{\lambda},\oriout{\mu}) \leq k$ and $m > i + \ell^2$.
    Recall that $i < m$ so $j_1 = \nextcut_\ell(i,\outp{\lambda})$ and $j_2 = \nextcut_\ell(i,\outp{\mu})$ exist.
    If $\maxdiff_{j_1,j_2}(\oriout{\lambda},\oriout{\mu}) > k$, \cref{rat:i1} is clearly satisfied.
    If $\maxdiff_{j_1,j_2}(\oriout{\lambda},\oriout{\mu}) \leq k$, we have to prove that $j_1 \neq j_2$ in order to satisfy \cref{rat:i1}.
    Towards a contradiction, assume $j_1 = j_2 =: j$. Since $i$ is maximal in $M$, we obtain that $j \ge m$.
    Let $\proot(\outp{\lambda}[i,j]) = u$ and $\proot(\outp{\mu}[i,j]) = v$ and $h_1,h_2$ such that $\outp{\lambda}[i,j] = u^{h_1}$ and $\outp{\mu}[i,j] = v^{h_2}$. 
    Because of the mismatch at $m \le j$, we have $u^{h_1} \neq v^{h_2}$. Since further $|u^{h_1}| = |v^{h_2}|$, the difference between $u^{h_1}, v^{h_2}$ must be on the first $|u||v|$ positions. This contradicts the fact that $m$ is the first mismatch between $\outp{\lambda}$ and $\outp{\mu}$ and $j \ge m \geq i + \ell^2$ and $|u||v| \leq \ell^2$.
    \end{proof}

  \subsection{Formal details for the regularity of the characterization}
  \label{sec:formaloverview}
  
  We provide here more details for the proof that $\cdelay{\subs}$ is regular (\cref{lem:regdashv}), and prove that it implies the regularity of $\mathbb{D}_{k,\ell,\subs}$ in \cref{lem:polRat}. First, we introduce some useful notions.
  
  Let $\Gamma$ be an arbitrary set and $\pi_\Sigma : (\Sigma\times \Gamma)^*\rightarrow \Sigma^*$ be the projection morphism which projects labels on $\Sigma$. We extend $\pi_\Sigma$ naturally to substitution sequences $\nu$ in $\sub{\var}{\Sigma\times \Gamma}^*$: $\pi_\Sigma$ projects on $\Sigma$ the labels from $\Sigma\times \Gamma$ occurring in the substitutions of $\nu$. In particular, $\outp{\pi_\Sigma(\nu)} = \pi_\Sigma(\outp{\nu})$.

  Let $n\geq 0$ and $u\in\Sigma^*$. For any $n$-tuple of positions
  $\overline{x}$ of $u$, we let
  $u\triangleright \overline{x}$ be the word in $(\Sigma\times
  2^{\{1,\dots,n\}})^*$
  of length $|u|$ such that for all $p\in\{1,\dots,|u|\}$, 
  $(u\triangleright \overline{x})[p] = (u[p],\{i\mid p=x_i\})$. An $n$-marking is any word of the form $u\triangleright \overline{x}$. E.g.,
  $abc\triangleright (1,3,1) = (a,\{1,3\})(b,\varnothing)(c,\{2\})$. 
  An \emph{$n$-marked substitution sequence} is a sequence of substitutions
  $\nu$ in $(\sub{\var}{\alp\times 2^{\{1,\dots,n\}}})^*$ such
  that $\outp{\nu}$ is an $n$-marking. Given $\lambda\in (\sub{\var}{\alp})^*$ and an $n$-tuple of positions $\overline{x}$ of $\outp{\lambda}$, we denote by $(\lambda\triangleright \overline{x})$ the unique $n$-marked substitution sequence such that $\outp{\lambda\triangleright \overline{x}} = \outp{\lambda}\triangleright \overline{x}$, $\pi_\Sigma(\lambda\triangleright \overline{x}) = \pi_\Sigma(\lambda)$ and each mark appears exactly once in $\lambda\triangleright \overline{x}$.
   For example, if $\var=\{X,O\}$ and $\lambda = ((X,O)\mapsto (Xc,Ob)).((X,O)\mapsto (\epsilon, aOX))$, then
   $\lambda\triangleright (1,3,1) = ((X,O)\mapsto (X(c,\{2\}),O(b,\varnothing))).((X,O)\mapsto (\epsilon, (a,\{1,3\})OX))$.

  
  
  
  \begin{restatable}{lemma}{lemCutNextCut}\label{coro:cutnextcutreg}
      The following sets of $1$- and $2$-marked
      substitution sequences $\lambda\triangleright i$ and 
  $\lambda\triangleright (i,j)$ where $\lambda\in\subs^*$ and $1\leq i\leq j\leq |\outp{\lambda}|$, are regular, recognizable by NFA with both $O(\ell|\Sigma|^{4|\var|\ell^3})$ states:
          \[
      \textsf{Cut}=  \{(\lambda\triangleright i)\mid
       i\in \cut_\ell(\outp{\lambda})\} \qquad \textsf{NextCut}  =  \{(\lambda\triangleright (i,j) \mid \nextcut_\ell(i,\outp{\lambda})=j\}
       \]
  \end{restatable}
  
  \begin{proof}[Proof sketch.]
  The proof is done in two steps:
  \begin{enumerate}
      \item we prove that $\{u \triangleright i\mid u\in\Sigma^*, i\in \cut_\ell(u)\}$ and 
      $\{u \triangleright (i,j)\mid u\in\Sigma^*, \nextcut_\ell(i,u)=j\}$ are recognizable by NFA with $O(\ell.|\Sigma|^{\ell+\ell^3})$ states, 
      \item then we transfer the latter result to marked substitution sequences.
  \end{enumerate}

      For step $1$, we actually prove a somewhat stronger result: the set of words $u$ such that
      all positions in $\cut_\ell(u)$ are marked by $1$ is
      regular, recognizable by an NFA $C$. Intuitively, the two main
      ingredients of the construction of $C$ are: (1) checking that a word is periodic
      with period \emph{at most} $\ell$, and (2) checking that a word
      is \emph{not} periodic for any period at most $\ell$.
      The first
      property is needed to check that words in between two marked
      positions are periodic of period $\ell'\leq \ell$, and it is doable by a finite automaton as
      $\ell$ is fixed (using $O(\ell.|\Sigma|^{\ell})$ states, where the
      factor $\ell$ comes from the fact that a period $\ell'\leq \ell$
      must be non-deterministically guessed). The second
      property is needed to check that the factor in between marked
      positions are maximal, in the sense that they cannot be extended
      to longer periodic factors of $u$, for any period at most $\ell$. This step
      is slightly more difficult but doable with
      $O(|\Sigma|^{\ell^3})$. The full proof is given in Appendix (\cref{lem:cutreg}). 
      
      For step $2$, we prove that given a regular language $P$ of $n$-markings of words over $\Sigma$, that we call regular predicate, the set of $n$-marked substitution sequences $\nu$ such that $\pi_\Sigma(\nu)\in \subs^*$ and $\outp{\nu}\in P$ is itself regular, recognizable by an NFA with $O(q^{2m})$ states, where $m = |\var|$ and $q$ is the number of states of an NFA recognizing $P$. This result is actually an easy consequence of the fact that given a
      regular language $L$ and some non-deterministic SST $T$, the inverse of $L$ by $T$
      is regular as well. It seems to be a folklore result but we could
      not find a reference where this result is explicitly mentioned
      with the
      right complexity, we reprove it in Appendix (\cref{lem:invreg}) for the sake of
      completeness. 
  \end{proof}
  
  We now prove some useful regularity results about substitutions
  sequences marked with positions satisfying some particular properties of their
  maximal difference $\maxdiff$.
  
  \begin{restatable}{lemma}{lemmaDelayReg}\label{lem:delaylargesmallreg}
      For all $\alpha\in\mathbbm{N}$, the following sets of pairs of $2$-marked substitution sequences $(\lambda\triangleright i_1)\otimes (\mu\triangleright i_2)$ where $\lambda,\mu\in\subs^*$, $1\leq i_1\leq \outp{\lambda}$ and $1\leq i_2\leq \outp{\mu}$ are regular and recognizable by NFA with $O(\alpha.8^{|\var|})$ states:
      \[
      \begin{array}{llllll}
      \mdiff_{\leq \alpha} & = & \{ (\lambda\triangleright i_1)\otimes (\mu\triangleright i_2)\mid
                    \maxdiff_{i_1,i_2}(\oriout{\lambda},\oriout{\mu})\leq
                             \alpha\} \\
       \mdiff_{>\alpha} & = & \{ (\lambda\triangleright i_1)\otimes (\mu\triangleright i_2)\mid
                  \maxdiff_{i_1,i_2}(\oriout{\lambda},\oriout{\mu})> \alpha\} \\
  
        \mdiffe_{\leq \alpha} & = & \mdiff_{\leq \alpha} \cap \{ (\lambda\triangleright i_1)\otimes (\mu\triangleright i_2)\mid
                     i_1=i_2\} \\ 
  
        \mdiffd_{\leq \alpha} & = & \mdiff_{\leq \alpha}\cap \{ (\lambda\triangleright i_1)\otimes (\mu\triangleright i_2)\mid
                     i_1\neq i_2\} \\ 
      \end{array}
       \]
  \end{restatable}
  
  \begin{proof}[Proof sketch.]
      Let us sketch the proof for $\mdiff_{\leq \alpha}$. We construct a finite
      automaton $A$ checking the conditions of
      the definition of $\mdiff_{\leq \alpha}$. For a substitution $s$ in
      $\lambda$ and an occurrence of a $\Sigma$-symbol in $s$, we say that
      this occurrence is \emph{relevant} if in the final output, it will
      occur before $i_1$. A similar notion can be defined for $\mu$. 
      When reading $(\lambda\triangleright i_1)\otimes
      (\mu\triangleright i_2)$, the automaton must count, incrementally, the number of
      relevant occurrences of $\Sigma$-symbols in substitutions of
      $\lambda \triangleright i_1$ minus the number
      of relevant occurrences of $\Sigma$-symbols in substitutions of
      $(\mu\triangleright i_2)$. Note that the absolute value of this difference after having
      read $t$ input pairs of substitutions is exactly
      $\diff_{i_1,i_2,t}(\oriout{\lambda},\oriout{\mu})$. The automaton updates this difference
      incrementally while reading pairs of substitutions, and as soon as
      its absolute value exceeds $\alpha$, the
      automaton rejects (or go to a sink accepting state if one wants to
      recognize $\mdiff_{>\alpha}$). It is correct since $\maxdiff$ is a max of
      $\diff$ over all input positions $t$, by definition. There is a
      little technicality due to the fact  it is unknown, when
      reading a pair of substitutions of $(\lambda\triangleright i_1)\otimes
      (\mu\triangleright i_2)$, which $\Sigma$-symbol occurrences are relevant and
      which are not. To overcome this, we first apply a rational relabelling into
      a more informative marking which also indicates which symbol occurrences are
      relevant or not, and obtain the result since rational
      relabelings preserve regularity under inverse. Note that if at
      the end of the computation, the difference is $0$ iff $i_1=i_2$,
      allowing to recognize the languages $\mdiffd_{\leq \alpha}$
      and $\mdiffe_{\leq \alpha}$. The proof is detailed in \cref{sec:proofLemDelayReg}. 
  \end{proof}
  
  Finally, the following result is useful to prove the regularity of the second property (\cref{rat:i2}) of the characterization of \cref{lem:carac}:
  
  \begin{restatable}{lemma}{lemmaMismatch}\label{lem:mismatch}
      The following set of pairs of $2$-marked substitution sequences $(\lambda\triangleright (i_1,j_1))\otimes (\mu\triangleright (i_2,j_2))$ where $\lambda,\mu\in \subs^*$ and $1\leq i_1,j_1\leq |\outp{\lambda}|$, $1\leq i_2,j_2\leq |\outp{\mu}|$, is regular, recognizable by an NFA with
                                $O((|\Sigma|\ell)^{8|\var|})$
                                 states:
      \[
         \textsf{Mism}_\ell = \{ (\lambda\triangleright
                                 (i_1,j_1))\otimes (\mu\triangleright
                                 (i_2,j_2))\mid 
                                                  0 \leq j_1-i_1 = j_2-i_2\leq \ell^2, 
                                               \outp{\lambda}[j_1]\neq\outp{\lambda}[j_2] \}
                                               \]                           
                             \end{restatable}

  \begin{proof}[Proof sketch.]
  As for \cref{coro:cutnextcutreg}, we first prove the regularity for the output language using the fact that the distance between $j_1$ and $i_1$ (and between $j_2$ and $i_2$) is bounded by a constant, and then transfer it to substitution sequences while preserving regularity. The full proof is given in \cref{sec:proofLemMismatch}.
  \end{proof}
  
  Finally we prove the main lemma of this section: regularity of the delay notion. 
  
  \begin{restatable}{lemma}{lemmapolRat}\label{lem:polRat}
      For every $k,\ell \in \N$ and any finite subset of substitutions
      $\subs$ over a set of variables $\var$ and alphabet $\Lambda$,
      the set $\mathbb{D}_{k,\ell,\subs}$ is regular,
      recognizable by a DFA with a number of states doubly exponential
      in $\ell^3$ and in $|\var|$, and singly exponential in $k$.
  \end{restatable}

  \begin{proof}[Proof sketch.]
     It is done in two steps: $(i)$ first, for any finite subset of substitutions $\subs'$, we prove the regularity of the set 
      $\cdelay{\subs'}$, and then $(ii)$ we show how to use it to prove \cref{lem:polRat}.

      For step $(i)$, the proof is based on the characterization of
      \cref{lem:carac}, which states sub-properties which can all be checked individually by the automata constructed in \cref{coro:cutnextcutreg,lem:delaylargesmallreg,lem:mismatch}. Taking a synchronized product of those automata yield the result. The technical details are in Appendix  (\cref{lem:regdashv}).

      For step $(ii)$, observe that the characterization of
      \cref{lem:carac},
  
      \noindent $(1)$ holds for pairs of substitution sequences whose output ends
      with some endmarker $\dashv$, and

      \noindent $(2)$ characterizes the pairs of such substitution sequences
      which are \emph{not} in $\mathbb{D}_{k,\ell,\subs}$.

      The restriction of observation $(1)$ can be removed as follows:
      first, we define $\subs = \subs' \cup \subs_\dashv$,
      where $\subs_\dashv$ is the set of substitutions of $\subs$ where
      an endmarker has been added in the update of the
      output variable $\varO$. The statement $(i)$ tells us that $\cdelay{\subs'}$ is regular. This yields regularity of 
      the complement of $\mathbb{D}_{k,\ell,\subs'}$, as noticed in the
      second observation),  complement which is restricted to substitutions
      sequences whose outputs end with an endmarker. The latter set is
      then complemented into a DFA. Finally, we apply a
      morphism which removes the endmarker to the substitutions of
      substitution sequences. This preserves the
      regularity (and the determinism of the automaton), as well as the delay. The full details are given in \cref{sec:proofLemPolRat}. 
  \end{proof}

    \subsection{Regularity of Cut and NextCut}
    
    \lemCutNextCut*
    
     The proof of this lemma is done in two steps, as explained in the main body of the paper. Those two steps are respectively given by \cref{lem:cutreg} and \cref{lem:invreg} below, whose combination immediately yields the desired result. In what follows, we call $n$-ary regular predicate any regular set of $n$-markings.

    \begin{restatable}{lemma}{lemmaCutReg}\label{lem:cutreg}
        The sets $\{u \triangleright i\mid u\in\Sigma^*, i\in \cut_\ell(u)\}$ and 
        $\{u \triangleright (i,j)\mid u\in\Sigma^*, \nextcut_\ell(i,u)=j\}$
        are unary and binary regular predicates respectively. They can be
        both recognized by NFA with $O(\ell.|\Sigma|^{\ell+\ell^3})$ states. 
    \end{restatable}

    \begin{proof}[Proof]
        For all $h\in\N$, let $P_{h}
        =\bigcup_{v\in \Sigma^{h}} v^*$ be the set of periodic words of
        period $h$. Note that $P_h$ is recognizable by a deterministic
        finite automaton $D_h$ with $O(|\Sigma|^h)$ states. 
        Now, for all words $u\in\Sigma^*$ and set $X$ of positions
        of $u$, let $u\triangleright X\in (\Sigma\times \{\varnothing,
        \{1\}\})^*$ be such that for all positions $i$, we have 
        $(u\triangleright X)[i] = (u[i],\{1\})$ if $i\in X$, otherwise it
        is $(u[i],\varnothing)$. We prove that the language
        $$
        L_{cut} = \{ u\triangleright \cut_\ell(u)\mid u\in\Sigma^*\}
        $$
        is regular. To that end, consider a word $u$ and an arbitrary
        subset $X = \{ i_1 < i_2 < \dots < i_n \}$ of positions of
        $u$. Let $i_0 = 0$.  It
        defines a factorization $u = u(i_0{:}i_1]u(i_1{:}i_2]\dots
        u(i_{n}{:}|u|]$ (note that the last factor is empty if $i_n =
        |u|$). We let $u_k = u(i_k{:}i_{k+1}]$.  An automaton $C$ recognizing $L_{cut}$ needs to check
        several properties in parallel (via a synchronized product construction):
    
        $(i)$ the maximal element of $X$ is $|u|$,

        $(ii)$ for all $0\leq k < n$, $u_k\in P_{\ell'}$ for some $1\leq
        \ell'\leq \ell$. To do so , $C$ successively for each $k$, non-deterministically guesses a period
        $\ell'$ and checks that $u_k$ is accepted by $D_{\ell'}$. In total, this needs $O(\ell.|\Sigma|^{\ell})$
        states.

        $(iii)$ for all $0\leq k < n$, the automaton $C$ must check that there is no factor
        of $u$ which starts at position $i_k+1$, belongs to some $P_{\ell'}$,
        and is strictly longer than $u_k$. In other words, $C$ must check that for
        all periods $1\leq \ell' \leq \ell$, if $\alpha_{\ell'}>i_{k+1}$ is the
        smallest integer such that $\alpha_{\ell'}-i_k\ mod\ \ell' = 0$, then either
        $\alpha_{\ell'}>|u|$ or  $u(i_k{:}\alpha_{\ell'}]\not\in
        P_{\ell'}$. To do so, for all $0\leq k < n$ and all $1\leq \ell'\leq \ell$, $C$ simulates a copy of the automaton $\overline{D_{\ell'}}$ (complement of
        $D_{\ell'}$) running on the factor $u(i_k{:}\alpha_{\ell'}]$ ($C$ also makes
        sure that $\alpha_{\ell'}$ is indeed the smallest integer such
        that $\alpha_{\ell'}-i_k\ mod\ \ell' = 0$ and $\alpha_{\ell'}\leq |u|$, by
        slightly modifying $\overline{D_{\ell'}}$). The different copies
        of the (modified) automata $\overline{D_{\ell'}}$ are run in
        parallel via a synchronized product, and must all accept. There is
        at most $\ell^2$ copies running in parallel, and each copy has
        $O(|\Sigma|^\ell)$ states, so, $C$ needs
        $O(|\Sigma|^{\ell^3})$ states for this verification.

        Overall, to check the three properties above in parallel, $C$ needs
        $O(\ell.|\Sigma|^{\ell+\ell^3})$ states. 
    
        It is not difficult, from an automaton $A$ recognizing $L_{cut}$,
        to define two automata recognizing the languages of the lemma
        statement. For the first language, it is a  kind of projection of
        $A$:  its
        transition on labels $(\sigma,\{1\})$ are duplicated as
        transitions on labels $(\sigma,\varnothing)$ and it is
        additionally checked that there exists exactly one position marked
        $\{1\}$ on the input.

        An automaton for the second language can be constructed using
        similar ideas.  
    \end{proof}

    \begin{restatable}{lemma}{lemmaInvReg}\label{lem:invreg}
        Let $n\geq 0$ be fixed. Given a regular $n$-ary predicate $P$ recognized by an NFA
        with $q$ states, the set of $n$-marked
        substitution sequences $\lambda$ over $\subs_n$ such that $\outp{\lambda}\in P$ is
        regular, recognized by an NFA with $O(q^{2m})$ states, where $m = |\var|$.
    \end{restatable}
    
    \begin{proof}
        Suppose that $T$ uses $m$ variables and $\alpha$
        states, and $L$ is defined by a NFA $A_L$ with $n_q$ states. The main
        idea is to construct an NFA $B$ which simulates $T$ and in parallel non-deterministically computes for each of its
        variables, a pair of states $(p,q)$, with the following invariant:
        for each variable $X$ and at each moment there is a run
        of $A_L$ from $p$ to $q$ on the content of $X$. So, states of $B$
        are of the form $Q_T\times (Var_T \rightarrow Q_L\times Q_L)$ where $Q_T$
        are the states of $T$, $Var_T$ its variables, and $Q_L$ the states
        of $A_L$. In a state $(p, f)$, reading a symbol $a$, the automaton
        $B$ goes to any state of the form $(p',f')$ such that $T$ goes from $p$ to $p'$ on
        $a$, producing a substitution $\sigma$, and $f'$ satisfies the
        following: for all $X\in Var_T$,  if $\sigma(X)$ is of the form 
        $u_1X_1u_2X_2\dots u_nX_nu_{n+1}$ where the $u_i$ do not contain
        variables, then there exists states $p_1,q_1,p_2,q_2,\dots,
        p_{n+1},q_{n+1}\in Q_L$ such that $f'(X) = (p_1,q_{n+1})$, $f(X_i)
        = (q_i, p_{i+1})$ and $p_i\xrightarrow{u_i}_{A_L} q_i$ for all
        $i=1\dots n$. The automaton accepts if at the end of the
        computation, the function $f$ computed satisfies that
        $f(\varO)\in I_L\times F_L$, where $\varO$ is the
        output variable of $T$, $I_L$ the set of initial states of $A_L$
        and $F_L$ its set of accepting states. So, the automaton $B$ has $|Q_L|^{2|Var_T|}$
        states.

        To get the statement of the lemma, it suffices to apply the latter
        result to a single initial-accepting state  SST $T$ whose input alphabet is
        $\subs_n$, with set of variables $\var$, and which does nothing
        but apply the substitution it reads as input. For $L$, it suffices
        to take the predicate $P$. 
        \end{proof}

      \subsection{Regularity results for max-diff}
      \label{sec:proofLemDelayReg}
    \lemmaDelayReg*
    
    \begin{proof}[Proof (formal details).]
    Before proving the lemma, we introduce some useful notions. Given a
    word $u\in\Sigma^*$, an \emph{$n$-labelling} of $u$ is
    a word in $v\in(\Sigma\times 2^{\{1,\dots,n\}})^*$ such that
    $\pi_1(v) = u$ and for all positions $p_1\leq p_2$ of $u$,
    $\pi_2(v[p_1])\supseteq \pi_2(v[p_2])$. For example, the word 
    $v = (a, \{1,2,3\})(b, \{1,2,3\})(c, \{1,3\})(d,\varnothing)$ is a
    $3$-labelling of $u = abcd$. Observe that there is a
    bijection between the $n$-labelings of $u$ and the
    $n$-markings of $u$: the $n$-marking of $u$ associated with 
    an $n$-labelling of $u$ is obtained by keeping only the right-most
    occurrence of $i\in\{1,\dots,n\}$, and conversely, any occurrence
    of $i$ in some $n$-marking is propagated leftward to all
    sets. This bijection from $n$-labelings to $n$-markings, denoted
    $\textsf{mk}$, is a rational function, defined by a finite transducer of
    constant size if $n$ is fixed. For example, we have $\textsf{mk}(v) = (a, \varnothing) (b,
    \{2\}) (c, \{1,3\}) (d, \varnothing)$.

    Similarly, an \emph{$n$-labelled substitution sequence} $\tau$ is a
    substitution sequence such that $\outp{\tau}$ is an
    $n$-labelling (of its projection onto $\Sigma$). Any $n$-labelled
    substitution sequence $\tau$ can be mapped injectively to an
    $n$-marked substitution sequence  we denote $\textsf{mks}(\tau)$ such
    that $\outp{\textsf{mks}(\tau)} =
    \textsf{mk}(\outp{\tau})$. The function $\textsf{mks}$ is
    defined by a finite transducer  $T$ with $O(2^{|\var|})$ states (if
    $n$ is fixed) we now describe. Suppose $T$  processes an $n$-labelled
    substitution sequence $\lambda$ and reads a substitution $s$
    of $\tau$. It must
    replace any set $I\subseteq \{1,\dots,n\}$ occurring in some
    update of some variable $X$, by $J\subseteq I$ such that $J$
    contains $i\in I$ if this is the last occurrence of $i$ in
    $\outp{\tau}$. To that end, the transducer guesses $J\subseteq
    I$ for all $j\in J$, verifies that $j$ will not occur again in
    $\outp{\tau}$. The latter requires $T$ to use non-determinism
    to guess the order in
    which the content of the variables in the substitutions
    occur in $\outp{\tau}$. This requires to keep $O(2^{|\var|})$ states.

    Now, observe that any $1$-labelled substitution sequence $\tau$
    uniquely defines a substitution sequence $\lambda$ and some
    position $i$ in $\outp{\lambda}$ such that $\textsf{mks}(\tau) = \lambda\triangleright
    i$. Let $G_\alpha = \{ \tau_1\otimes \tau_2\mid \exists
    \lambda.\textsf{mks}(\tau_1) = \lambda\triangleright i_1,
    \textsf{mks}(\tau_2) = \lambda\triangleright i_2,
    \maxdiff_{i_1,i_2}(\oriout{\lambda},\oriout{\mu})\leq \alpha\}$. We have
    $$
    \mdiff_{\leq \alpha} = \{ \textsf{mks}(\tau_1)\otimes \textsf{mks}(\tau_2)
    \mid \tau_1\otimes \tau_2\in G_\alpha\}
    $$
    Therefore, if $G_\alpha$ is regular, definable by an NFA with $q$
    states, $\mdiff_{\leq \alpha}$ is regular and definable by an NFA with
    $O(q.2^{|\var|})$ states, because $\textsf{mks}$ is a
    letter-to-letter rational function defined by a finite transducer
    with $O(2^{|\var|})$ states.

    We show that $G_\alpha$ is regular, definable by an NFA $A$ with
    $O(\alpha.2^{|\var|})$ states.  The automaton $A$ is obtained via a synchronized
    product of three finite automata which respectively check that:
    \begin{enumerate}
      \item for each pair of substitutions $(\sigma_1,\sigma_2)$ read
        as input, their $\Sigma$-projection are equal, i.e.,
        $\pi_\Sigma(\sigma_1) = \pi_\Sigma(\sigma_2)$. This ensures that $\textsf{mks}(\tau_1) =
    \lambda\triangleright i_1$ and
    $\textsf{mks}(\tau_2) = \lambda \triangleright i_2$ for the same
    $\lambda$, and some $i_1,i_2$. Only one state is needed.

  \item it reads pairs of $1$-labelled substitution sequences as
    inputs. The set $L_1$ of words $v$ such
    that $v$ is a $1$-labelling of some word $u\in\Sigma^*$ is easily
    seen to be regular, definable by some NFA with $2$ states.
    We can then apply \cref{lem:invreg} (for $n=0$)
    and $P = L_1$, and get that using $O(4^{|\var|})$ states, $A$ can
    verify that its input consists in pairs of $1$-labelled
    substitution sequences. 

  \item $\maxdiff_{i_1,i_2}(\oriout{\lambda},\oriout{\mu})\leq \alpha$: to do so, it needs
    $2.\alpha+2$  states in
    the $D = \{-\alpha,-\alpha+1,\dots,\alpha-1,\alpha,\infty\}$ to count
    the differences of the weights in the interval
    $[-\alpha,\alpha]\cap \mathbbm{N}$ (and state $\infty$ if the
    difference exits that interval). Let $(s_1,s_2)$ be a pair of substitutions and $d\in D$. The
    state $\infty$  is a non-accepting sink state while the other states
    from $D$ are accepting. State $0$ is the unique initial state. 
    Suppose
    that $d\neq \infty$. Let $n_1$
    be the number of symbols from $\Sigma\times \{\{1\}\}$ occurring
    in the right hand-sides of $s_1$, and $n_2$ be defined similarly
    but for $s_2$. When reading $(s_1,s_2)$ from state $d$, $A$ goes
    to $\infty$ if $d+n_1-n_2 < -\alpha$ or $d+n_1-n_2>\alpha$,
    otherwise it goes to state $d+n_1-n_2$. 
\end{enumerate}
So, overall, $A$ has $O(\alpha.8^{|\var|})$ states.

To recognize $\mdiffe_{\leq \alpha}$, we just modify the third automaton: all states of
$D$ are rejecting but state $0$. Indeed, $i_1=i_2$ iff
$\maxdiff_{i_1,i_2}(\oriout{\lambda},\oriout{\mu})=0$. So, if the automaton does not reach the
sink state $\infty$ and ends up in state $0$, we are guaranteed that
$(i)$ the maximal absolute difference is lesser than $\alpha$ and
$(ii)$ $i_1=i_2$. For $\mdiffd_{\leq \alpha}$, it suffices to
set states $0,\infty$ to be rejecting and all states in $D\backslash
\{0,\infty\}$ to be accepting. For $\mdiffe_{> \alpha}$, only
$\infty$ is set to be accepting.
\end{proof}

  \subsection{Regularity of mismatches}
  \label{sec:proofLemMismatch}
  \lemmaMismatch*
  \begin{proof}[Proof.]    
  
      We prove that the language
  \[
  H = \{ (u\triangleright
                                 (i_1,i_1+d))\otimes (v\triangleright
                                 (i_2,i_2+d))\mid \lambda,\mu\in\subs^*,
                                 0\leq d\leq \ell^2,
                                 u[i_1+d]\neq v[i_2+d] \}
                                 \]
      is recognizable by an NFA with $O(|\Sigma|^2\ell^4)$ states using the fact that $\ell$ is a given constant. Then, it suffices to apply
      \cref{lem:invreg} to conclude.

      For all $\sigma\in \Sigma$
      and $0\leq d\leq \ell^2$, 
      let $P_{\sigma,d}$ be the following binary predicate:
        \[
        P_{\sigma,d} = \{ u\triangleright (i,i+d)\mid u\in\Sigma^*,
        1\leq i\leq i+d\leq |u|,\ u[i+d]=\sigma\}
       \]
        This predicate is easily seen to be regular, recognized by a DFA
        $A_{\sigma,d}$ with $O(d)$ states. The operator $\otimes$ is naturally extended
        to languages, and we get:
        \[
        H = \bigcup_{\sigma_1,\sigma_2\in \Sigma, \sigma_1\neq \sigma_2}
      \bigcup_{0\leq d\leq \ell^2} P_{\sigma_1,d}\otimes P_{\sigma_2,d}
       \]
        Moreover, for all $\sigma_1,\sigma_2,d$, one can construct a DFA
        recognizing $P_{\sigma_1,d}\otimes P_{\sigma_2,d}$ with $O(d^2)$
        states, and take the union of all these DFA to obtain an NFA
        recognizing $H$ with $O(|\Sigma|^2\ell^4)$ states. 
        \end{proof}

          \subsection{Final steps: Regularity of the delay notion}
          \label{sec:proofLemPolRat}
    
          In order to prove \cref{lem:polRat}, we need one more auxiliary lemma.
    
        \begin{restatable}{lemma}{lemRegDashv}\label{lem:regdashv}
        Let $\subs$ be a finite set of substitutions
        over a finite set of variables $\var$ and an alphabet
        $\Sigma$. Let $\dashv\in\Sigma$ and let $T_\dashv = \{
        \lambda\in\subs^*\mid \outp{\lambda}\in
        (\Sigma\backslash\{\dashv\})^*{\dashv}\}$. Let
        $k,\ell\in\N$. Then, the set 
        \[
        \cdelay{\subs} = \{ \lambda\otimes \mu\mid \lambda,\mu\in T_\dashv\text{ and }
        (\lambda,\mu)\not\in\mathbb{D}_{k,\ell,\subs}\}
        \]
        is regular, recognized by an NFA with a number of states exponential
        in $\ell^3$ and in $|\var|$, and polynomial in $k$.
    \end{restatable}
    
    \begin{proof}
     
        First, for $x\in\{1,\dots,5\}$, we define
        the set $C_x$ of convolutions of marked
        substitution sequences $(\lambda\triangleright (i,j_1))\otimes
        (\mu\triangleright (i,j_2))$ satisfying
        the $x$th property below:
        \begin{enumerate}
        \item $i \in \left(\cut_\ell(\outp{\lambda}) \cap
              \cut_\ell(\outp{\mu})\right)$ and $\maxdiff_i(\oriout{\lambda},\oriout{\mu})\leq k$
        \item $j_1 = \nextcut_\ell(i,\outp{\lambda})$, $j_2 =
          \nextcut_\ell(i,\outp{\mu})$
          \item $\maxdiff_{j_1,j_2}(\oriout{\lambda},\oriout{\mu}) >~k$
    
              \item $\maxdiff_{j_1,j_2}(\oriout{\lambda},\oriout{\mu}) \leq~k$ and $j_1=
                j_2$
              \item $\outp{\lambda}[j_1] \neq \outp{\mu}[j_2]$,
                $j_1-i\leq \ell^2$ and $j_2-i\leq \ell^2$ .
            \end{enumerate}

            Similarly, we define languages corresponding to the various
            properties of \cref{lem:carac} but in the case $i=0$. So,  
            for $x\in\{1,\dots,4\}$, we define
            the set $Z_x$ of convolutions of marked
            substitution sequences $(\lambda\triangleright j_1)\otimes
            (\mu\triangleright j_2)$ satisfying the
            $x$th property below:
            \begin{enumerate}
              \item $j_1 = \nextcut_\ell(0,\outp{\lambda})$, $j_2 =
                \nextcut_\ell(0,\outp{\mu})$
              \item $\maxdiff_{j_1,j_2}(\oriout{\lambda},\oriout{\mu}) >~k$
                
              \item $\maxdiff_{j_1,j_2}(\oriout{\lambda},\oriout{\mu}) \leq~k$ and $j_1=
                j_2$
              \item $\outp{\lambda}[j_1] \neq \outp{\mu}[j_2]$,
                $j_1\leq \ell^2$ and $j_2\leq \ell^2$ .
            \end{enumerate}
    
            Let $\pi_\Sigma$ be the projection of marked substitutions on
            their $\Sigma$ component. So, $\pi_\Sigma$ substitutes any
            label $(\alpha, M)$ occurring in the right-hand side of a substitution, where $\alpha\in\Sigma$ and $M$ is a
            marking, by just $\alpha$.  According to \cref{lem:carac}, 
            we have:
            \[
            \cdelay{\subs} = \pi_\Sigma(T_\dashv\cap
            (C_1\cap C_2 \cap C_3 \cup C_1\cap C_2\cap C_4\cup C_1\cap
            C_5\cup Z_1\cap Z_2\cup Z_1\cap Z_3\cup Z_4))
            \]
            Since $\pi_\Sigma$ preserves regular languages (in linear time
            for regular sets given as NFA), and regular languages are
            closed under Boolean operation, it suffices to prove that each
            set in the above expression is regular. For
            $C_1,C_2,C_3,C_4$, it is an immediate consequence of
            the various results we have previously shown:
            \cref{coro:cutnextcutreg},
            \cref{lem:delaylargesmallreg} and
            \cref{lem:mismatch}. For example, for $C_1$, it suffices
            to simulate twice the automaton recognizing the language $Cut$ of
            \cref{coro:cutnextcutreg}, on inputs of the form
            $(\lambda\triangleright (i_1,j_1))\otimes (\mu\triangle
            (i_2,j_2))$, and intersect it with the automaton recognizing
            $\textsf{MD}^=_{\leq k}$ of
            \cref{lem:delaylargesmallreg}. More precisely, using a
            product construction:
            \begin{itemize}
                \item the automaton for $Cut$ is simulated
            both twice on $(\lambda\triangleright (i_1,j_1))\otimes (\mu\triangle
            (i_2,j_2))$: the first simulation ignores all components but
            the one corresponding to $\lambda$ and the marking of $i_1$,
            while the other ignores all components but the one
            corresponding to $\mu$ and the marking of $i_2$. 
          \item the automaton for $\textsf{MD}^=_{\leq k}$ is simulated by
            ignoring the markings corresponding to $j_1$ and~$j_2$. 
            \end{itemize}
           This yields an automaton for $C_1$ with
           $O(k8^{|\var|}\ell^2|\Sigma|^{4|\var|\ell^3})$ states.

           The regularity of  $Z_1,\dots,Z_4$ can be shown using 
           similar ideas. Observe they are particular cases of
           $C_2,\dots,C_5$ for $i=0$. Since $i$ is not a position of the
           input, it cannot be marked. However, modulo adding a starting
           symbol, say, $\vdash$, occurring uniquely, the regularity of
           $C_2,\dots,C_5$ and this starting symbol removed
           while preserving regularity.

           Based on the number of states of each automata
           obtained in \cref{coro:cutnextcutreg},
            \cref{lem:delaylargesmallreg} and
            \cref{lem:mismatch}, a careful analysis yields an NFA for
            $\mathcal{C}^{\dashv}_{k,\ell,\subs}$ with a number of states
            which is:
            \begin{itemize}
              \item exponential
                in $\ell^3$
              \item exponential in $|\var|$
              \item polynomial in $k$.
            \end{itemize}
    \end{proof}

    \lemmapolRat*
    \begin{proof}[Proof (formal details).]
        
        Let $k,\ell \in \N,\var,\Lambda,\subs$ as in the statement of the
        lemma. Let $\dashv\not\in \Lambda$ and $\alpo =
        \Lambda\cup\{\dashv\}$. Let $\Phi : \sub{\var}{\Lambda}\rightarrow
        \sub{\var}{\alpo}$ defined for all $s\in \sub{\var}{\Lambda}$ by
        $\Phi(s)(X) = s(X)$ if $X\neq \varO$ and
        $\Phi(s)(\varO) = s(X).{\dashv}$. Note that $\Phi$ is a
        bijection. We let $\subs' = \subs\cup \Phi(\subs)$. Note that
        $|\subs'| = 2|\subs|$. We let $\Psi : \subs' \rightarrow \subs$
        defined by $\Psi(s) = s$ if $s\in \subs$ and $\Psi(s) =
        \Phi^{-1}(s)$ if $s\in \Phi(\subs)$. So, $\Psi$ removes the
        occurrence of the endmarker, if it exists.

    
    \cref{lem:regdashv} precisely proves that the following set is regular:
        \[
        \cdelay{\subs'} = \{ \lambda\otimes \mu\mid \lambda,\mu\in T_\dashv\text{ and }
        (\lambda,\mu)\not\in\mathbb{D}_{k,\ell,\subs'}\}
        \]
        From that, we obtain that
        the following set:
        \[
        \mathbb{D}_{k,\ell,\subs'}^{\dashv} = \{ \lambda\otimes \mu\mid
        \lambda,\mu\in T_\dashv\text{ and } (\lambda,\mu)\in\mathbb{D}_{k,\ell,\subs'}\}
        \]
        is regular. Indeed, we have 
        \[
        \mathbb{D}_{k,\ell,\subs'}^{\dashv} = ((\subs'\times
        \subs')^*\setminus \cdelay{\subs'})\cap (T_\dashv\otimes
        T_\dashv)
        \]
        where here $\otimes$ is extended to languages. Since $\cdelay{\subs'}$ is
        recognizable by an NFA with a number of states which is
        exponential in $\ell^3$ and $|\var|$, and polynomial in $k$, 
    $((\subs'\times
        \subs')^*\setminus \cdelay{\subs'})$ can be
        recognized by a DFA $D_1$ with a number of states doubly exponential in
        both $\ell^3$ and $|\var|$, and exponential in $k$. The set $T_\dashv$ is recognizable by a DFA $D_2$ with
        $4^{|\var|}$ states by an immediate application of
        \cref{lem:invreg} to $2$-states automaton accepting words with a
        unique occurrence of $\dashv$ at the end.
        From the DFA
        $D_2$, it is easy to get a DFA $D_3$ recognizing $T_\dashv\otimes
        T_\dashv$, with $4^{2|\var|}$ states. Taking the
        product of $D_1$ and $D_3$, one obtains a DFA recognizing 
        $\mathbb{D}_{k,\ell,\subs'}^{\dashv}$ with a number of states doubly exponential in
        $\ell^3$ and in $|\var|$, and exponential in $k$.

        Let us now explain how to obtain a DFA for
        $\mathbb{D}_{k,\ell,\subs}$. The mapping $\Psi$ which removes occurrences of
        endmarkers in substitutions is first extended two pairs of
        substitutions by $\Psi(s,s') = (\Psi(s),\Psi(s'))$ and then
        morphically to sequences of the form $\lambda\otimes \mu$. 
        Then, we claim that
        \[
        \mathbb{D}_{k,\ell,\subs} =
        \Psi(\mathbb{D}^{\dashv}_{k,\ell,\subs'})\qquad\qquad\qquad (\star)
        \]
        Indeed, let $\lambda\otimes \mu \in
        \mathbb{D}_{k,\ell,\subs}$. Let $\lambda = \lambda_1 s$ for
        $\lambda_1\in \subs^*$ and $s\in\subs$. Similarly, decompose $\mu$
        into $\mu = \mu_1 t$. Then, let $\lambda' = \lambda_1.\Phi(s)$ and
        $\mu' = \mu_1.\Phi(t)$. We have $\Psi(\lambda'\otimes \mu') =
        \lambda\otimes \mu$. Moreover, $\lambda',\mu'\in T_\dashv$,
        $\outp{\lambda'} = \outp{\lambda}.{\dashv} = \outp{\mu}.{\dashv} =
        \outp{\mu'}$, because $\outp{\lambda} = \outp{\mu}$ since
        $\lambda\otimes\mu\in \mathbb{D}_{k,\ell,\subs}$. 
        It can be seen that
        $d_\ell(\lambda',\mu') = d_\ell(\lambda,\mu)$ (adding an
        endmarker does not modify the delay). So, $\lambda'\otimes \mu'\in
        \mathbb{D}^{\dashv}_{k,\ell,\subs'}$. We have shown inclusion
        $\subseteq$ in equality  $(\star)$. The other inclusion is shown
        similarly: removing the endmarker from two equal outputs does not
        modify the delay.

        To conclude,  the mapping $\Psi$ is an alphabetic morphism, so, by the
        regularity of $\mathbb{D}^{\dashv}_{k,\ell,\subs'}$, we also get
        the regularity of $\mathbb{D}_{k,\ell,\subs}$, and its
        recognizability by an DFA of the same size as the DFA recognizing
        $\mathbb{D}^{\dashv}_{k,\ell,\subs'}$: doubly exponential in
        $\ell^3$ and in $|\var|$, and exponential in $k$. 
        \end{proof}

%% file: Appendix/inclusion.tex
    \section{Inclusion and equivalence up to fixed delay} 
      
    \thmComplexityInclusion*
    
    \begin{proof}
    To begin with, we introduce some notation.
    For a non-deterministic SST $T$ over input alphabet $\Sigma$, we denote by $L(T)$ its language,
    defined as the set of words of the form $u\otimes \tau(\rho)$, where
    $u\in \Sigma^*$, $\rho$ is an accepting run of $T$ over $u$, and
    $\tau(\rho)$ is the sequence of substitutions occurring on $\rho$. 
    Let $T_1$ and $T_2$ be two non-deterministic SSTs over two finite sets of variables
    $\var_1$ and $\var_2$ respectively, both with output variable $\varO$. Let $\subs_1$ (resp. $\subs_2$) be the finite set of
    substitutions occurring in $T_1$ (resp. $T_2$). Let $\subs =
    \subs_1\cup \subs_2$ and $\var = \var_1\cup \var_2$. Let
    $\ell,k\in\N$.
    Let $n_1$ and $n_2$ be the number of states of $T_1$ and $T_2$, respectively.
        
        We prove the result for the inclusion problem, i.e. $T_1\subseteq_{k,\ell} T_2$. We let
        $\mathbb{D}^{in}_{k,\ell,\subs} = \{ u \otimes \lambda \otimes
        \mu\mid u\in\Sigma^*\wedge \lambda \otimes \mu\in \mathbb{D}_{k,\ell,\subs}\}$.  
        The automaton recognizing $\mathbb{D}_{k,\ell,\subs}$ in
        \cref{lem:polRat} can be easily extended into a DFA
        $A_D$ which recognizes
        $\mathbb{D}^{in}_{k,\ell,\subs}$, using the same number of states,
      i.e., doubly exponential in $\ell^3$ and $|\var|$, and singly
      exponential in $k$.

      Now, observe that $T_1$
      is $(k,\ell)$-included in $T_2$ iff $L(T_1)\subseteq
      \mathbb{D}^{in}_{k,\ell,\subs}(L(T_2))$, where $\mathbb{D}^{in}_{k,\ell,\subs}(L(T_2))$ denotes the set $\{ u \otimes \lambda \mid \exists u \otimes \mu \in L(T_2)\colon u \otimes \lambda \otimes \mu \in \mathbb{D}^{in}_{k,\ell,\subs}\}$. Let $A_1$ be an NFA recognizing $L(T_1)$. It is easily obtained by simulating
    $T_1$ and has $n_1$ states, where $n_1$ is the number of states of
    $T_1$. Similarly, let $A_2$ be an NFA recognizing $L(T_2)$ with $n_2$ states.

    Clearly, $\mathbb{D}^{in}_{k,\ell,\subs}(L(T_2))$ is definable by
    restricting the first and third component of $A_D$ to $L(T_2)$ and then projecting onto the first and second component.
    This can be done via a product construction with $A_2$. 
    This yields a NFA $B$ for $\mathbb{D}^{in}_{k,\ell,\subs}(L(T_2))$ with doubly exponentially many states in $\ell^3$ and $|\var|$, singly exponentially many states in $k$, and polynomially many  states in $n_2$. 
    Note the $B$ has polynomially many transitions in the size of $T_1$ and size of $T_2$, because the number of transitions is polynomial in $|\subs|$ which depends on the size of $T_1$ and~$T_2$.
    
    So, the problem $T_1\subseteq_{k,\ell} T_2$ boils down to checking the inclusion of the NFA $A_1$ of polynomial size in $n_1$ into the NFA $B$, which is well-known to be decidable in \textsc{PSpace} in the size of both NFA. This yields a \textsc{2ExpSpace} procedure.
    
    Now, if $k,\ell$ and $|\var|$ are fixed constants, note that $A_D$ has a constant number of states and a polynomial number of transitions in the size of $T_1$ and $T_2$, because it depends on~$\subs$. Thus, $B$ has a polynomial number of states in $n_2$ and polynomially many transitions. Therefore, the latter procedure is $\textsc{PSpace}$.

    The \PSPACE lower bound is a direct consequence of the fact
    that the language inclusion problem for NFA is
    \PSPACE-hard. This problem can be encoded to the
    $(0,0)$-inclusion problem for non-deterministic SSTs which uses only a
    single variable $\varO$ which is always updated by $\varO
    := \epsilon$. So, those SSTs do not produce anything and are basically
    NFA.
    \end{proof}